\documentclass[english]{article}
\usepackage{geometry}
\usepackage{float}
\usepackage{mathrsfs}
\usepackage{amsmath}
\usepackage{amssymb}
\usepackage{graphicx}
\usepackage{amsfonts}
\usepackage{bm}
\usepackage{subfigure}
\usepackage{amsthm}
\usepackage{algorithm}
\usepackage{algorithmic}
\usepackage{etoolbox}
\usepackage{booktabs}
\usepackage{soul}
\usepackage{bbm}
\usepackage{epsfig}
\makeatletter

\@ifundefined{definecolor}{\@ifundefined{definecolor}
 {\@ifundefined{definecolor}
 {\usepackage{color}}{}
}{}
}{}
\newcommand{\red}[1]{{#1}}

\usepackage{subfig}\usepackage[all]{xy}

\newcommand{\bbR}{\mathbb R}

\newcommand{\bbone}{\mathbbm 1}

\newtheorem{theorem}{Theorem}[section]
\newtheorem{lem}{Lemma}[section]
\newtheorem{rem}{Remark}[section]
\newtheorem{proposition}[theorem]{Proposition}

\newtheorem{assumption}{Assumption}[section]

\newtheorem{remark}[theorem]{Remark}
\newcounter{hypA}

\newcommand{\cA}{\mathcal{A}}
\newcommand{\cC}{\mathcal{C}}
\newcommand{\bbE}{\mathbb{E}}
\newcommand{\cO}{\mathcal{O}}
\newcommand{\bbT}{\mathbb{T}}
\newcommand{\bbP}{\mathbb{P}}

\usepackage{babel}\date{}

\begin{document}

\begin{center}

{\Large \textbf{Multilevel Particle Filters for L\'{e}vy-driven stochastic differential equations }}

\bigskip

BY AJAY JASRA $^{1}$, KODY J. H. LAW $^{2}$ \& PRINCE PEPRAH OSEI $^{1}$ 

$^{1}$ {\footnotesize Department of Statistics \& Applied Probability,
	National University of Singapore, Singapore, 117546, SG.}\\
{\footnotesize E-Mail:\,}\texttt{\emph{\footnotesize op.peprah@u.nus.edu, staja@nus.edu.sg}}\\
$^{2}$ {\footnotesize School of Mathematics, University of Manchester, UK, AND Computer Science and Mathematics Division, Oak Ridge National Laboratory Oak Ridge, TN, 37831, USA.}\\
{\footnotesize E-Mail:\,}\texttt{\emph{\footnotesize kodylaw@gmail.com}}
\end{center}

\begin{abstract}
We develop algorithms for computing expectations  
with respect to the laws of
 models associated to
 stochastic differential equations (SDEs) 
driven by pure L\'{e}vy processes.  We consider filtering such processes and well as
pricing of path dependent options.
We propose a multilevel particle filter (MLPF) to address the computational issues involved in solving these continuum problems.  
We show via numerical simulations and theoretical results that under suitable assumptions regarding the discretization 
of the underlying driving L\'{e}vy proccess, our proposed method achieves optimal convergence rates:
the cost to obtain MSE $\mathcal{O}(\epsilon^2)$ scales like $\mathcal{O}(\epsilon^{-2})$ 
for our method, as compared with the standard particle filter $\mathcal{O}(\epsilon^{-3})$. 
\end{abstract}

\section{Introduction}\label{levy:intro} 
L\'{e}vy processes have become very useful recently in several scientific disciplines.  A non-exhaustive list 
includes physics, in the study of turbulence and quantum field theory; economics, for continuous time-series models; insurance mathematics, for computation of insurance and risk, and 
mathematical finance, for pricing path dependent options.  Earlier application of L\'{e}vy processes in modeling financial instruments dates back in \cite{Madan_VGprocess} where a variance gamma process is used to model market returns.  

A typical computational problem in mathematical finance is the computation of the quantity 
$\mathbb{E}\left[f(Y_t)\right]$, where $Y_t$ is the time $t$ 
solution of a stochastic differential equation driven by a L\'{e}vy process and $f\in\mathcal{B}_b(\mathbb{R}^d)$, 
a bounded Borel measurable function on $\mathbb{R}^d$. 
For instance $f$ can be a payoff function. 
Typically one uses  the Black-Scholes model, 
in which the underlying price process is lognormal.  However, often 
the asset price exhibits big jumps over the time horizon.  
The inconsistency of the 
assumptions of the Black-Scholes model for market data has 
lead to the development of more realistic models for these data in the literature.  
General L\'{e}vy processes offer a promising 
alternative to describe the observed reality of financial market data, 
as compared to models that are based on standard Brownian motions.

In the application of standard and multilevel particle filter methods to SDEs driven by general L\'{e}vy processes, 
in addition to pricing path dependent options, we will consider filtering of partially-observed L\'{e}vy process with discrete-time observations.  
In the latter context, we will assume that the partially-observed data are regularly spaced observations 
$z_1,\dots,z_n$, where $z_k\in\mathbb{R}^d$ is a realization of $Z_k$ and $Z_k|(Y_{k\tau}=y_{k\tau}$) has density given by 
$g\left(z_k|y_{k\tau}\right)$, where $\tau$ is the time scale.  
Real 
S\&P $500$ stock price data will be used to illustrate our proposed methods 
as well as the standard particle filter.  
We will show how both of these problems can be formulated as general
Feynman-Kac type problems \cite{delm:04}, with time-dependent potential functions
modifying the L{\'e}vy path measure.

The multilevel Monte Carlo (MLMC) 
methodology was introduced in \cite{heinrich2001multilevel}
and first applied to the simulation of SDE driven by Brownian motion in \cite{Giles_mlmc}. 
Recently, \cite{Dereich_mlmcLevydriven} provided a detailed analysis of the application of MLMC to a L\'{e}vy-driven SDE. 
{This first work was extended in \cite{Dereich_gausscorrection} to a method with a Gaussian correction term which can 
substantially improve the rate for pure jump processes \cite{Asmusen_Rosinski_levyprocesses}.}
The authors in \cite{castilla_mlmcLevy} use the MLMC method for general L\'{e}vy processes based on Wiener-Hopf decomposition.  
We extend the methodology described in \cite{Dereich_mlmcLevydriven} to a particle filtering framework.  
This is challenging due to the following reasons.  
First, 
one must choose a suitable weighting function to prevent the weights in the particle filter being zero (or infinite).  
Next, one must 
control the jump part of the underlying L\'{e}vy process such that the path of the filter does not blow up as the time parameter increases.  In pricing path dependent options, for example knock out barrier options, we adopt the same strategy described in \cite{JayOption, Jay_smcdiff} for the computation of the expectation of the functionals of the SDE driven by general L\'{e}vy processes.

The rest of the paper is organised as follows.  In Section \ref{levy:PODLevy}, we briefly review the construction of general L\'{e}vy processes,  
the numerical approximation of L\'{e}vy-driven SDEs, 
the MLMC method, 
and finally the construction of a coupled kernel for L\'{e}vy-driven SDEs which will allow MLMC to be used.  
Section \ref{levy:ML_Levy} introduces both the standard and multilevel particle filter methods and their application to L\'{e}vy-driven SDEs.  
Section \ref{levy:numerics} features numerical examples of pricing barrier options and filtering of partially observed L\'{e}vy processes. 
The computational savings of the multilevel particle filter over the standard particle filter is illustrated in this section.  

\section{
Approximating SDE driven by 
L\'{e}vy Processes}\label{levy:PODLevy}

In this section, we briefly describe the construction and approximation 
of a general $d^{\prime}$-dimensional L\'{e}vy process $\{X_t\}_{t\in[0,K]}$,
and the solution $Y:=\{Y_t\}_{t\in[0,K]}$ of a $d$-dimensional SDE driven by $X$.
Consider a stochastic differential equation given by
\begin{align}\label{levy:eq1}
\mathrm{d}Y_t&=a(Y_{t^{-}})\mathrm{d}X_t,\quad \mathrm{y}_0\in\mathbb{R}^d,
\end{align}
where 
$a:\mathbb{R}^d\rightarrow\mathbb{R}^{d\times d^{\prime}}$, 
and the initial value is $\mathrm{y}_0$ (assumed known).  
In particular, in the present work we are interested in computing the expectation of bounded and measurable functions 
$f:\mathbb{R}^d\rightarrow\mathbb{R}$, 
that is $\mathbb{E}[f(Y_t)]$. 


\subsection{L\'{e}vy Processes}\label{levy:levproc}

For a general detailed description of the L\'{e}vy processes and analysis of SDEs driven by L\'{e}vy processes, 
we shall refer the reader to the monographs of \cite{Bertoin_levyprocesses,Sato_levyprocesses} and \cite{Applebaum_levyprocesses,Protter_Isde}.  
L\'{e}vy processes are stochastic processes with stationary and independent increments, which begin almost surely from  the origin and are stochastically continuous.  Two important fundamental tools available to study the richness of the class of  L\'{e}vy processes are the L\'{e}vy-Khintchine formula and the L\'{e}vy-It\^{o} decomposition.  They respectively characterize the distributional properties and the structure of sample paths of the L\'{e}vy process.  Important examples of L\'{e}vy processes include Poisson processes, compound Poisson processes and Brownian motions.

There is a strong interplay between L\'{e}vy processes and infinitely divisible distributions such that, for any $t>0$ the distribution of $X_t$ is infinitely divisible.  Conversely, if $F$ is an infinitely divisible distribution then there exists a L\'{e}vy process $X_t$ such that the distribution of $X_1$ is given by $F$.  This conclusion is the result of L\'{e}vy-Khintchine formula for L\'{e}vy processes we describe below.  Let $X$ be a L\'{e}vy process with a triplet 
$\left(\nu,\Sigma,b\right)$, $b\in\mathbb{R}^{d'}, 0\leq\Sigma=\Sigma^T \in\mathbb{R}^{d'\times d'}$, 
where $\nu$ is a measure 
satisfying $\nu(\{0\})=0$ and $\int_{\mathbb{R}^{d'}}(1\wedge|x|^2)\nu(\mathrm{d}x)<\infty$, 
such that
\begin{align*}
	\mathbb{E}[e^{i\langle u, X_t\rangle }]&=\int_{\mathbb{R}^{d'}}e^{i\langle u, x\rangle }\pi(\mathrm{d}x)=e^{t\psi(u)}
\end{align*}
with $\pi$ the probability law of $X_t$, where
\begin{align}\label{levy:eq2}
\psi(u)&=i\langle u, b\rangle-\frac{\langle u, \Sigma u \rangle}{2}+
\int_{\mathbb{R}^{d'}\backslash\{0\}}\left(e^{i\langle u , x\rangle }-1-i\langle u , x\rangle 
\right)\nu(dx),\quad u\in\mathbb{R}^{d'}.
\end{align}
The measure $\nu$ is called the L\'{e}vy measure of $X$.  
The triplet of L\'{e}vy characteristics $\left(\nu,\Sigma,b\right)$ is simply called L\'{e}vy triplet.  
Note that in general, the L\'{e}vy measure $\nu$ can be finite or infinite.  
If $\nu(\bbR)<\infty$, then almost all paths of the L\'{e}vy process have a finite number of jumps on every compact interval {and it can be represented as a compensated compound Poisson process.}  
On the other hand, if $\nu(\bbR)=\infty$, 
then the process has an infinite number of jumps on every compact interval almost surely.
{Even in this case the third term in the integrand ensures that the integral is finite,
and hence so is the characteristic exponent.}

\subsection{Simulation of L\'{e}vy Processes}\label{levy:levsimulation}

The law of increments of many L\'{e}vy processes is not known explicitly.  
This makes it more difficult to simulate a path of a general L\'{e}vy process 
than for instance standard Brownian motion.  
For a few L\'{e}vy processes where the distribution of the process is known explicitly, \cite{Cont_jumpprocesses,Schoutens_levyprocfinance} provided 
methods for exact simulation of 
such processes, which are applicable in financial modelling.  
For our purposes, the simulation of the path of a general L\'{e}vy process will be based on the L\'{e}vy-It\^{o} decomposition and we briefly describe the construction below.  
An alternative construction is based on Wiener-Hopf decomposition.
This is used in \cite{castilla_mlmcLevy}. 

The L\'{e}vy-It\^{o} decomposition reveals much about the structure of the paths of a L\'{e}vy process.  We can split the L\'{e}vy exponent, or the characteristic exponent of $X_t$ in $\left(\ref{levy:eq2}\right)$, into three 
parts
\begin{align*}
\psi&=\psi^{1}+\psi^2+\psi^3 \, . 
\end{align*}
where 
\begin{align*}
\psi^1(u)&=i\langle u , b\rangle ,\quad \psi^2(u)=-\frac{\langle u, \Sigma u \rangle}{2}, \\ 
\psi^3(u)&=\int_{\mathbb{R}^{d'}\backslash\{0\}}\left(e^{i\langle u , x\rangle }-1-i\langle u , x\rangle 
\right)\nu(dx),\quad u\in\mathbb{R}^{d'}
\end{align*}
The first term corresponds to a deterministic drift process with parameter $b$, 
the second term to a Wiener process with covariance 
${\sqrt{\Sigma}}$, {where $\sqrt{\Sigma}$ denotes the symmetric square-root}, and 
the last part corresponds to 
a L\'{e}vy process which is a square integrable martingale. 
This term may either be a compensated compound Poisson process or the limit of such processes, 
and it is the hardest to handle when it arises from such a limit.  

Thus, any L\'{e}vy process can be decomposed into three independent L\'{e}vy processes thanks to the L\'{e}vy-It\^{o} decomposition theorem.
In particular, let $\{W_t\}_{t\in[0,K]}$ denote a Wiener process independent of the 
process $\{L_t\}_{t\in[0,K]}$.  
A L\'{e}vy process $\{X_t\}_{t\in[0,K]}$ can be {\emph constructed} as follows
\begin{align}\label{levy:eq3}
X_t&=\sqrt{\Sigma} W_t+L_t+bt \, .
\end{align}
The L\'{e}vy-It\^{o} decomposition guarantees that every square integrable L\'{e}vy process has a representation as $\left(\ref{levy:eq3}\right)$.  
We will assume that one cannot sample from the law of $X_t$, hence of $Y_t$, 
and rather we must numerically approximate the process with finite resolution.
Such numerical methods have been studied extensively, for example in \cite{Jacod_levydrivensde,Rubenthaler_levyprocess}.

It will be assumed that the L{\'e}vy process $X$ \eqref{levy:eq2}, 
and the L{\'e}vy-driven process $Y$ in \eqref{levy:eq1}, 
satisfy the following conditions.  Let $|\cdot|$ denote the standard Euclidean $l2$ norm, 
for vectors, and induced operator norm for matrices.

\begin{assumption} \label{ass:main}
There exists a $C>0$ such that 
\begin{itemize}
\item[{\rm (i)}] $|a(y) - a(y')| \leq C |y-y'|$, and $|a(y)| \leq C$ for all $y\in \bbR^d$ ;
\item[{\rm (ii)}] $0 < \int |x|^2 \nu(dx) \leq C^2$ ; 
\item[{\rm (iii)}] $|\Sigma|< C^2$ and $|b| \leq C$ \, .
\end{itemize}
\end{assumption}

Item (i) provides continuity of the forward map, while (ii) controls the variance of the jumps, and (iii) controls the diffusion and drift components and is trivially satisfied.  
These assumptions are the same as in the paper \cite{Dereich_mlmcLevydriven}, with the exception of the second part of (i), which was not required there.
As in that paper we refer to the following general references on L{\'e}vy processes for further details\cite{Applebaum_levyprocesses, Bertoin_levyprocesses}.

\subsection{Numerical Approximation of a L\'{e}vy Process and L\'{e}vy-driven SDE}
\label{numApprox}

Recall $\left(\ref{levy:eq1}\right)$ and $\left(\ref{levy:eq3}\right)$. 
Consider the evolution of discretized L\'{e}vy process and hence the L\'{e}vy-driven SDE over the time interval $[0,K]$.

In order to describe the Euler discretization of the two processes for a given accuracy parameter $h_l$, 
we need some definitions.  The meaning of the subscript will become clear in the next section. 
Let $\delta_l>0$ denote a 
jump threshold parameter 
in the sense that jumps which are smaller than $\delta_l$ will be ignored.
Let $B_{\delta_l}=\{x\in\mathbb{R}^{d'}:|x|<\delta_l\}$.  
Define $\lambda_l=\nu(B_{\delta_l}^{c})<\infty$, that is the L\'{e}vy 
measure outside of the ball of radius $\delta_l$.  
We assume that the L\'evy component of the process is 
nontrivial so that $\nu(B_1)=\infty$.
First $h_l$ will be chosen and then 
the parameter $\delta_l$ will be chosen such that the step-size of the time-stepping method is 
$h_l=1/\lambda_l$.  
The jump time increments are exponentially distributed with parameter $\lambda_l$ so that 
the number of jumps before time $t$ is a Poisson process $N^l(t)$
with intensity $\lambda_l$.
The jump times will be denoted by $\tilde{T}_j^l$.
The jump heights $\Delta L_{\tilde{T}_j}^l$ 
are distributed according to
\begin{align*}
\mu^l(\mathrm{d}x)&:=\frac{1}{\lambda_l}\bbone_{B_{\delta_l}^{c}}(x)\nu(\mathrm{d}x).
\end{align*} 
Define 
\begin{equation}\label{eq:efl}
F_0^l=\int_{B_{\delta_{l}}^{c}}x\nu(\mathrm{d}x).
\end{equation}
The expected number of jumps on an interval of length $t$ is $F_0^l t$, 
and \red{the compensated compound Poisson process $L^{\delta}$ defined by
$$L_t^{\delta}=\sum_{j=1}^{N^l(t)} \Delta L_{\tilde{T}_j}^l - F_0^l t$$ 
is an $L^2$ martingale which 
converges in $L^2$ to the L{\'e}vy process $L$ 
as $\delta_l \rightarrow 0$ \cite{Applebaum_levyprocesses,Dereich_mlmcLevydriven}}.

The Euler discretization of the L\'{e}vy process and the L\'{e}vy driven SDE
is given by 
Algorithm \ref{levy:DiscreteAlgo}.  
Appropriate refinement 
of the original jump times $\{\tilde{T}^l_j\}$ to new jump times $\{T_{j}^{l}\}$ 
is necessary to control the discretization error arising from 
the Brownian motion component, the original drift process, and the drift 
component of the compound Poisson process.  
Note that  
the $\Delta L_{{T}^l_j}^{l}$ is non-zero only when $T_{j}^{l}$ corresponds with $\tilde{T}_{m}^{l}$ for some $m$,
as a consequence of the construction presented above.

\begin{algorithm}[!ht]
	\caption{\textbf{: Discretization of L\'{e}vy process}} 
	\label{levy:DiscreteAlgo}
	\begin{algorithmic}
		\STATE
		Initialization: Let $\tilde{T}_0^l=0$ and $j=1$;
		\begin{enumerate}
			\item[(A)]  Generate jump times:
%
%
				$\tilde{T}_j^l=\min\{1,\tilde{T}_{j-1}^l+\xi_j^l\}$, $\xi_j^l\sim Exp(\lambda_l)$\, ;
				
				If $\tilde{T}_j^l=1, \tilde{k}_l=j$; Go to (B);
				
				Otherwise $j=j+1$;
				Go to start of (A).
			
			\item[(B)]  Generate jump heights:
			
			For $j\in\{1,\dots,\tilde{k}_l-1\}$, $z_j^l\sim\mu^l$;
			
			$\Delta L_{\tilde{T}_j^l}^l=z_j^l$ and $\Delta L_{\tilde{T}_{\tilde{k}_l}}^l=0$;

						Set $j=1$, $T_0^l=0$.

			\item[(C)]  Refinement of original jump times:
			
				$T_j^l=\min\bigg\{T_{j-1}^l+h_l, \min\Big\{\tilde{T}_k^l > T_{j-1}^l;k\in\{1,\dots,\tilde{k}_l\}\Big\}\bigg\}$ \, ;
				
				If $T_j^l=\tilde{T}_k^l$ for some $k\in \{1,\dots, \tilde{k}_l\}$, 
				then $\Delta L_{T^l_j} = \Delta L_{\tilde{T}^l_j}$; otherwise $\Delta L_{T^l_j} =0$ \, ;
				
				If $T_j^l=1, k_l=j$; Go to (D);
				
				
				
				Otherwise $j=j+1$; Go to start of (C). 
		\end{enumerate}
	\end{algorithmic}
\end{algorithm}

\begin{algorithm}[!ht]
	\begin{algorithmic}
		\STATE
		\begin{enumerate}
			\item[(D)]  Recursion of the process:
			
			For $m\in\{0,\dots,k_l-1\}$, $X_0^l=x_0$; 
			\begin{equation}\label{levy:eq4}
			X^l_{T^l_{m+1}} =X^l_{T^l_{m}}+\sqrt{\Sigma}\Big(W_{T^l_{m+1}}-W_{T^l_{m}}\Big)+
			\Delta L_{{T}_{m+1}^l}^l 
			+(b-F_0^l)(T^l_{m+1}-T^l_{m}) \, .
			\end{equation}
		\end{enumerate}
	\end{algorithmic}
\end{algorithm}

The numerical approximation of the L\'{e}vy process described in Algorithm \ref{levy:DiscreteAlgo}
gives rise to an approximation of 
the L\'{e}vy-driven SDE as follows.  
Given $Y^l_{T^l_{0}}$, for $m=0,\dots, k_l-1$
\begin{equation}\label{eq:euler_levyd}
			Y^l_{T^l_{m+1}} =Y^l_{T^l_{m}}+a(Y^l_{T^l_{m}})(\Delta X)^l_{T^l_{m+1}},
\end{equation}
where $(\Delta X)^l_{T^l_{m+1}} = X^l_{T^l_{m+1}}-X^l_{T^l_{m}}$ is given by \eqref{levy:eq4}.
In particular the recursion in \eqref{eq:euler_levyd} gives rise to a 
transition kernel, denoted by 
$Q^l(u,dy)$, 
between observation times $t\in\{0,1,\dots,K\}$.  
This kernel is the measure of 
$Y^l_{T^l_{k_l}}$
given initial condition 
$Y^l_{T^l_0}=u$.  
Observe that the initial condition for $X$ is irrelevant for simulation of $Y$, 
since only the increments $(\Delta X)^l_{T^l_{m+1}}$ are required, 
which are simulated independently by adding a realization of 
$N\big ((b-F_0^l)(T^l_{m+1}-T^l_{m}), (T^l_{m+1}-T^l_{m})\Sigma \big)$
to $\Delta L_{T_{m+1}^l}^l$.


\begin{rem}
	The numerical approximation of the L\'{e}vy process and hence L\'{e}vy-driven SDE $\left(\ref{levy:eq1}\right)$ in Algorithm \ref{levy:DiscreteAlgo} is the single-level version of a more general coupled discretization \cite{Dereich_mlmcLevydriven} which will be described shortly in Section \ref{levy:couplekernel}.  This procedure will be used to obtain samples for the plain particle filter algorithm.  
\end{rem}

\subsection{Multilevel Monte Carlo Method} 
\label{levy:mlmc}

Suppose one aims to approximate the expectation of functionals of the solution of the
L\'{e}vy-driven SDE in $\left(\ref{levy:eq1}\right)$ at time $1$, that is $\mathbb{E}[f(Y_1)]$, 
where $f:\mathbb{R}^d\rightarrow\mathbb{R}$ is a bounded and measurable function.  
Typically, one is interested in the expectation w.r.t. the law of exact solution of SDE $\left(\ref{levy:eq1}\right)$, 
but this is not always possible in practice. 
Suppose that the law associated with $\left(\ref{levy:eq1}\right)$ with no discretization is $\pi_1$.
Since we cannot sample from $\pi_1$, we use a biased version $\pi^L_1$ associated with a given level of discretization of SDE $\left(\ref{levy:eq1}\right)$ at time $1$.   
Given $L\geq1$, define $\pi_{1}^{L}(f):=\mathbb{E}[f(Y^L_{1})]$, the expectation with respect to the density associated with the Euler discretization $\left(\ref{levy:eq4}\right)$ at level $L$.  
The standard Monte Carlo (MC) approximation at time $1$ consists in obtaining i.i.d. samples $\Big(Y_{1}^{L,(i)}\Big)_{i=1}^{N_L}$ from the density $\pi_1^L$ and approximating $\pi_{1}^{L}(f)$ by its empirical average
\begin{align*}
{\pi}_1^{L,N_L}(f)&:=\frac{1}{N_L}\sum_{i=1}^{N_L}f(Y_1^{L,(i)}).
\end{align*}
The mean square error of 
the estimator 
is
\begin{align*}
e({\pi}_1^{L,N_L}(f))^2&:=\mathbb{E}\left[\left({\pi}_1^{L,N}(f)-\pi_{1}(f)\right)^2\right].
\end{align*}
Since the MC estimator ${\pi}_1^{L,N_L}(f)$ is an unbiased estimator for $\pi_{1}^{L}(f)$, 
 this can further be decomposed into
\begin{equation}\label{eq:error_dec}
e({\pi}_1^{L,N_L}(f))^2=\underbrace{N_L^{-1}\mathbb{V}[f(Y_1^L)]}_{\hbox{variance}}
+(\underbrace{{\pi}_1^L(f)-\pi_{1}(f)}_{\hbox{bias}})^2.
\end{equation}
The first term in the right hand side of the decomposition is the variance of MC simulation and the second term is the bias 
arising from discretization. 
If we want \eqref{eq:error_dec} to be $\mathcal{O}(\epsilon^2)$, 
then it is clearly necessary to choose $N_L\propto \epsilon^{-2}$, 
and then the total cost is $N_L\times {\rm Cost}(Y_1^{L,(i)}) \propto \epsilon^{-2 - \gamma}$, 
where  it is assumed that ${\rm Cost}(Y_1^{L,(i)}) \propto \epsilon^{-\gamma}$ for some $\gamma>0$ 
is the cost to ensure the bias is $\mathcal{O}(\epsilon)$.

Now, in the multilevel Monte Carlo (MLMC) settings, one can observe that the expectation of the finest approximation ${\pi}_1^L(f)$ can be written as a telescopic sum starting from a coarser approximation ${\pi}_{1}^{0}(f)$, and the intermediate ones:
\begin{align}\label{levy:eq5}
{\pi}_{1}^L(f)&:={\pi}_{1}^{0}(f)+\sum_{l=1}^{L}\left({\pi}_{1}^{l}(f)-{\pi}_{1}^{l-1}(f)\right).
\end{align}
Now it is our hope that the variance of the increments decays with $l$, 
which is reasonable in the present scenario where they are finite resolution approximations of a limiting process.
The idea of the MLMC method is to approximate the multilevel (ML) identity $\left(\ref{levy:eq5}\right)$ 
by independently computing each of the expectations in the telescopic sum by a standard MC method.  
This is possible by obtaining i.i.d. 
pairs of samples $\Big(Y_{1}^{l,(i)},Y_{1}^{l-1,(i)}\Big)_{i=1}^{N_l}$ for each $l$,
from a suitably coupled joint measure $\bar{\pi}_1^l$ with the appropriate marginals $\pi_1^l$ and $\pi_1^{l-1}$, 
for example generated from a coupled 
simulation the Euler discretization of SDE $\left(\ref{levy:eq1}\right)$ at successive refinements.  
The construction of such a coupled kernel 
is detailed in Section \ref{levy:couplekernel}.  
Suppose it is possible to obtain such coupled samples at time $1$.
Then 
for $l=0,\dots,L$, one has independent MC estimates.  Let
\begin{align}\label{levy:eq6}
{\pi}^{{N}_{0:L}}_{1}(f)&:=\frac{1}{N_0}\sum_{i=1}^{N_0}f(Y_1^{1,(i)})+\sum_{l=1}^{L}\frac{1}{N_l}\sum_{i=1}^{N_l}\left(f(Y_1^{l,(i)})-f(Y_1^{l-1,(i)})\right),
\end{align}
where ${N}_{0:L}:=\left\lbrace N_l\right\rbrace_{l=0}^{L}$.  Analogously to the single level Monte Carlo method, the mean square error for the multilevel estimator $\left(\ref{levy:eq6}\right)$ can be expanded to obtain
\begin{align}\label{levy:eq7}
e\left({\pi}_{1}^{{N}_{0:L}}(f)\right)^2&:=\underbrace{\sum_{l=0}^{L}N_l^{-1}\mathbb{V}[f(Y_1^{l})-f(Y_1^{l-1}))]}_{\hbox{variance}}+(\underbrace{{\pi}_1^L(f)-\pi_{1}(f)}_{\hbox{bias}})^2,
\end{align}
with the convention that $f(Y_1^{-1})\equiv 0$.  
It is observed that the bias term remains the same; that is we have not introduced any additional bias.  However, by an optimal choice of $N_{0:L}$, one can possibly reduce the computational cost for any pre-selected tolerance of the variance of the estimator, or conversely reduce the variance of the estimator for a given computational effort.

In particular, for a given user specified error tolerance $\epsilon$ measured in the root mean square error, the highest level $L$ and the replication numbers $N_{0:L}$ are derived as follows.  We make the following assumptions about the bias, variance and computational cost based on the observation that there is an exponential decay of bias and variance as $L$ increases.

Suppose that there exist some constants $\alpha,\beta,\gamma$ and an accuracy parameter $h_l$ associated with the discretization of SDE $\left(\ref{levy:eq1}\right)$ at level $l$ such that
\begin{itemize}
	\item[$\left(B_l\right)$] $|\mathbb{E}[f(Y^{l})-f(Y^{l-1})]|=\mathcal{O}(h_{l}^{\alpha})$,
	\item[$\left(V_l\right)$] $\mathbb{E}[|f(Y^{l})-f(Y^{l-1})|^2]=\mathcal{O}(h_{l}^{\beta})$,
	\item[$\left(C_l\right)$] $\hbox{cost}\left(Y^{l},Y^{l-1}\right) \propto h_{l}^{-\gamma}$,
\end{itemize}
where $\alpha,\beta,\gamma$ are related to the particular choice of the discretization method 
and cost is the computational effort to obtain one sample $\left(Y^l,Y^{l-1}\right)$.  
For example, the Euler-Maruyama discretization method for the solution of SDEs driven by Brownian motion 
gives orders $\alpha =\beta =\gamma=1$.  
The accuracy parameter $h_l$ typically takes the form $h_l=S_{0}^{-l}$ for some integer $S_{0}\in\mathbb{N}$.
Such estimates can be obtained for L\'evy driven SDE and this point will be revisited in detail below.
For the time being we take this as an assumption.

The key observation from the mean-square error of the multilevel estimator $\left(\ref{levy:eq6}\right)-\left(\ref{levy:eq7}\right)$ is that the bias is given by  the finest level, while the variance is decomposed into a sum of variances of the  $l^{th}$ increments.  
Thus the total variance is of the form $\mathcal{V}=\sum_{l=0}^{L}V_lN_l^{-1}$ and by condition  $\left(V_l\right)$ above, the variance of the $l^{th}$ increment is of the form $V_lN_l^{-1}$.  
The total computational cost takes the form $\mathcal{C}=\sum_{l=0}^{L}C_lN_l$.
In order to minimize the effort to obtain a given mean square error (MSE), one must balance the terms in $\left(\ref{levy:eq7}\right)$.  Based on the condition $\left(B_l\right)$ above, a bias error proportional to $\epsilon$ will require the highest level

\begin{align}\label{levy:eq8}
L&\propto\frac{-\log(\epsilon)}{\log(S_{0})\alpha}.
\end{align}

In order to obtain optimal allocation of resources $N_{0:L}$, one needs to solve a constrained optimization problem: 
minimize the total cost $\mathcal{C}=\sum_{l=0}^{L}C_lN_l$ for a given fixed total variance $\mathcal{V}=\sum_{l=0}^{L}V_lN_l^{-1}$ or vice versa.  
Based on the conditions $\left(V_l\right)$ and $\left(C_l\right)$ above, one obtains via the Lagrange multiplier method
the optimal allocation $N_l\propto V_l^{1/2}C_l^{-1/2}
\propto h_{l}^{(\beta+\gamma)/2}$.

Now targetting an error of size $\mathcal{O}(\epsilon)$, one sets $N_l\propto\epsilon^{-2}h_{l}^{(\beta+\gamma)/2}K(\epsilon)$, where $K(\epsilon)$ is chosen to control the total error for increasing $L$.  Thus, for the multilevel estimator we obtained:
\begin{align*}
\hbox{variance}&:\mathcal{V}=\sum_{l=0}^{L}V_lN_l^{-1}=\epsilon^{2}K(\epsilon)^{-1}\sum_{l=0}^{L}h_{l}^{(\beta-\gamma)/2}\\
\hbox{cost}:&\thickspace\mathcal{C}=\sum_{l=0}^{L}C_lN_l=\epsilon^{-2}K(\epsilon)^{2}.
\end{align*}
One then sets $K(\epsilon)=\sum_{l=0}^{L}h_{l}^{(\beta-\gamma)/2}$ in order to have variance of $\mathcal{O}(\epsilon^2)$.  We can identify three distinct cases 
\begin{itemize}
	\item[(i).] If $\beta=\gamma$, which corresponds to the Euler-Maruyama scheme, then $K(\epsilon)=L$.  One can clearly see from the expression in $\left(\ref{levy:eq8}\right)$ that $L=\mathcal{O}(|\log(\epsilon)|)$.  Then the total cost is $\mathcal{O}(\epsilon^{-2}\log(\epsilon)^2)$ compared with single level $\mathcal{O}(\epsilon^{-3})$.
	\item[(ii).] If $\beta > \gamma$, which correspond to the Milstein scheme, then $K(\epsilon)\equiv 1$, and hence  the optimal computational cost is $\mathcal{O}(\epsilon^{-2})$.
	\item[(iii).] If $\beta < \gamma$, which is the worst case scenario, 
	then 
	it is sufficient to choose $K(\epsilon)=K_{L}(\epsilon)=
	h_{L}^{(\beta-\gamma)/2}$.  In this scenario, one can easily deduce that
	the total cost is $\mathcal{O}(\epsilon^{-(\gamma/\alpha+\kappa)})$, where $\kappa=2-\beta/\alpha$, 
	using the fact that $h_L\propto \epsilon^{1/\alpha}$.
\end{itemize}

One of the defining features of the multilevel method is that the realizations $(Y_1^l,Y_1^{l-1})$
for a given increment must be sufficiently coupled in order to obtain decaying variances $(V_l)$.  
It is clear how to accomplish this in the context of stochastic differential equations driven by Brownian motion introduced in \cite{Giles_mlmc} (see also \cite{Jay_mlpf}), where coarse icrements are obtained by summing the fine increments, but it is non-trivial how to proceed in the context of SDEs purely driven by general L\'{e}vy processes.  
A 
technique based on {Poisson thinning} has been suggested by 
\cite{GilesXia_mlmcjumpdiff} for pure-jump diffusion and by \cite{castilla_mlmcLevy} for general L\'{e}vy processes.  
In the next section, we explain an alternative 
construction of a coupled kernel based on the L{\'e}vy-Ito decomposition, 
in the same spirit as in \cite{Dereich_mlmcLevydriven}.

\subsection{Coupled Sampling for Levy-driven SDEs}\label{levy:couplekernel}

The ML methodology described in Section \ref{levy:mlmc} works by obtaining samples from some coupled-kernel associated with discretization of $\left(\ref{levy:eq1}\right)$.  We now describe how one can construct such a kernel associated with the discretization of 
the L\'{e}vy-driven SDE.  
Let $u=(y,y')\in \bbR^{2d}$.
Define a kernel, $M^l:[\mathbb{R}^d\times\mathbb{R}^d]\times[\sigma(\mathbb{R}^d)\times\sigma(\mathbb{R}^d)]\rightarrow \mathbb{R}_{+}$, where $\sigma(.)$ denotes the $\sigma$-algebra of measurable subsets, such that for $A\in \sigma(\mathbb{R}^d)$
\begin{eqnarray}\label{eq:em}
M^l(u,A) &:=& M^{l}(u,A\times\mathbb{R}^d) =\int_{A}Q^l\left(y,\mathrm{d}z\right)=Q^l(y,A), \\
\label{eq:em2}
M^{l-1}(u,A) &:=& M^{l}(u,\mathbb{R}^d\times A) =\int_{A}Q^{l-1}\left(y',\mathrm{d}z\right)=Q^{l-1}(y',A).
\end{eqnarray}

The coupled kernel $M^l$ can be constructed using the following strategy.  Using the same definitions in Section \ref{numApprox}, let $\delta_l$ and $\delta_{l-1}$ be user specified jump-thresholds for the fine and coarse approximation, respectively.  Define 
\begin{equation}\label{eq:effs}
F_0^l=\int_{B_{\delta_{l}}^{c}}x\nu(\mathrm{d}x) \quad {\rm and} \quad
F_0^{l-1}=\int_{B_{\delta_{l-1}}^{c}}x\nu(\mathrm{d}x).
\end{equation}
The 
objective is to generate a coupled pair $(Y_1^{l,l},Y_1^{l,l-1})$ given $(Y_0^{l},Y_0^{l-1})$, $h_l,h_{l-1}$
with $h_l<h_{l-1}$. 
The parameter $\delta_\ell(h_\ell)$ will be chosen such that 
$h_\ell^{-1}=\nu(B_{\delta_\ell}^c)$, 
and these determine the value of $F_0^\ell$ in \eqref{eq:effs}, for $\ell\in\{l,l-1\}$.
We now describe the construction of the coupled kernel $M^l$ and thus obtain 
the coupled pair in Algorithm \ref{levy:coupledkernelAlgo}, 
which is the same as the one presented in \cite{Dereich_mlmcLevydriven}.

\begin{algorithm}[!ht]
\caption{\textbf{: Coupled kernel $M^l$ for 
L\'{e}vy-driven SDE}}
\label{levy:coupledkernelAlgo}
\begin{algorithmic}
\STATE
\begin{enumerate}	
\item[$(1)$]  Generate fine process:
Use parts (A) to (C) of Algorithm \ref{levy:DiscreteAlgo} to generate fine process yielding			
$\Big(\Delta L_{{T}_1^{l,l}}^{l,l},\dots, \Delta L_{{T}_{{k}_l^l}^{l,l}}^{l,l}\Big)$ 
and $\Big(T_1^{l,l},\dots,{T}_{k_l^l}^{l,l}\Big)$

\item[$(2)$]  Generate coarse jump times and heights: for $j_l\in\{1,\dots,{k}_l^l\}$	,		
	
	If  $\Delta L_{{T}_{j_l}^{l,l}}^{l,l}\geq\delta_{l-1}$, 
then	
	$\Delta L_{\tilde{T}_{j_{l-1}}^{l,l-1}}^{l,l-1}=\Delta L_{{T}_{j_l}^{l,l}}^{l,l}$ and
	$\tilde{T}_{j_{l-1}}^{l,l-1}={T}_{j_l}^{l,l}$;		$j_{l-1}=j_{l-1}+1$;
			
				
%
%
%
%
%
%
%
					

\item[$(3)$]  
Refine jump times: Set $j_{l-1}=j_l=1$ and $T_0^{l,l-1}=\overline{T}^{l,l}_0=0$,

(i) 
$T_{j_{l-1}}^{l,l-1}=\min\bigg\{T_{j_{l-1}-1}^{l,l-1}+h_{l-1}, 
\min\Big\{\tilde{T}_k^{l,l-1}\geq T_{j_{l-1}-1}^{l,l-1};k\in\{1,\dots,\tilde{k}_{l-1}^{l}\}\Big\}\bigg\}$.

If $T_{j_{l-1}}^{l,l-1}=1$, set $k_{l-1}^l=j_{l-1}$; else $j_{l-1}=j_{l-1}+1$ and Go to (i).

(ii) 
$\overline{T}_{j_l}^{l,l}=\min\bigg\{T \geq \overline{T}_{j_l-1}^{l,l}; T \in \{T_k^{l,l-1}\}_{k=1}^{k_{l-1}^l} 
\cup \{T_k^{l,l}\}_{k=1}^{k_{l}^l} \bigg \}$. 

If $\overline{T}_{j_l}^{l,l}=1$, set $k_l^l=j_l$, and redefine $T_i^{l,l} := \overline{T}_i^{l,l}$ for $i=1,\dots, k_l^l$;

Else $j_l=j_l+1$ and Go to (ii).

%
%
%
\item[$(4)$]  Recursion of the process: sample $W_{T_1^{l,l}},\dots W_{T_{k_l^l}^{l,l}}$ 
	(noting 
	$\{T_{k}^{l,l-1}\}_{k=1}
	^{k_{l-1}^l} 
	\subset\{T_{k}^{l,l}\}_{k=1}^
	{k_l^l}$); 
	
		Let $m_l\in\{0,\dots,k_l^l-1\}$, 
				 $m_{l-1}\in\{0,\dots,k_{l-1}^l-1\}$, 
				 $Y_0^{l,l}=Y_0^l$ , and $Y_0^{l,l-1}=Y_0^{l-1}$;
\end{enumerate}				 
\vspace{-20pt}
	\begin{eqnarray}\label{levy:eq9}
Y^{l,l}_{T^{l,l}_{m_l+1}}&=&Y^{l,l}_{T^{l,l}_{m_l}}+
a\Big(Y^{l,l}_{T^{l,l}_{m_l}}\Big)\Big(\sqrt{\Sigma}\Delta W_{T^{l,l}_{m_l+1}} 
+ \Delta L_{{T}_{m_l+1}^{l,l}}^{l,l} + (b-F_0^l)\Delta T^{l,l}_{m_l+1} \Big ) \, , \\ 
%
\label{levy:eq10}
Y^{l,l-1}_{T^{l,l-1}_{m_{l-1}+1}}&=&Y^{l,l-1}_{T^{l,l-1}_{m_{l-1}}}+a\Big(Y^{l,l-1}_{T^{l,l-1}_{m_{l-1}}}\Big)
\Big(\sqrt{\Sigma}\Delta W_{T^{l,l-1}_{m_{l-1}}} + \Delta L_{{T}_{m_{l-1}+1}^{l,l-1}}^{l,l-1} 
+ (b-F_0^{l-1})\Delta T^{l,l-1}_{m_{l-1}}\Big) \, ,
\end{eqnarray}
where $\Delta W_{T^{l,\ell}_{m_\ell+1}}= W_{T^{l,\ell}_{m_\ell+1}}-W_{T^{l,\ell}_{m_\ell}}$ 
and $\Delta T^{l,\ell}_{m_\ell+1} = T^{l,\ell}_{m_\ell+1}-T^{l,\ell}_{m_\ell}$, for $\ell\in\{l,l-1\}$.
\end{algorithmic}
\end{algorithm}


The construction of the coupled kernel $M^l$ outlined in Algorithm \ref{levy:coupledkernelAlgo} ensures that the paths of fine and coarse processes are correlated enough to ensure that the optimal convergence rate of the multilevel algorithm is achieved.

\section{Multilevel Particle Filter for L\'{e}vy-driven SDEs}\label{levy:ML_Levy}

In this section, the multilevel particle filter will be discussed for sampling from certain types of measures which have a density with respect to a L{\'e}vy process.
We will begin by briefly reviewing the general framework 
and standard particle filter, and then we will extend these ideas into the multilevel particle filtering framework.

\subsection{Filtering and Normalizing Constant Estimation for L\'{e}vy-driven SDEs}\label{levy:FilterLevySDE}

Recall the L\'{e}vy-driven SDE \eqref{levy:eq1}.  We will use the following notation here $y_{1:n}=[y_1, y_2,\dots,y_n]$.
It will be assumed that the general probability density of interest is of the form for $n\geq 1$, for some given $y_0$
\begin{align}\label{eq:target}
\hat{\eta}^{\infty}_{n}(y_{1:n}) \propto \Big[\prod_{i=1}^{n}G_i(y_{i})Q^{\infty}(y_{i-1}, y_i)\Big], 
\end{align}
where $Q^{\infty}(y_{i-1},y)$ is the transition density of the process $\left(\ref{levy:eq1}\right)$ as a function of $y$, i.e. the density of solution $Y_1$ at observational time point $1$ given initial condition $Y_0=y_{i-1}$. It is assumed that $G_i(y_{i})$ is the conditional density  (given $y_i$) of an observation at discrete time $i$, so observations (which are omitted from our notations) are regularly observed at times $1,2,\dots$. Note that the formulation discussed here, that is for $\hat{\eta}^{\infty}_{n}$, also allows one to consider general Feynman-Kac models (of the form \eqref{eq:target}), rather than just the filters that are focussed upon in this section.
The following assumptions will be made on the likelihood functions $\{G_i\}$. Note these assumptions are needed for our later mathematical
results and do not preclude the application of the algorithm to be described.
\begin{assumption}
There are $c>1$ and $C>0$, 
such that for all $n>0$, and 
$v, v' \in \bbR^{d}$, 
$G_n$ satisfies
\begin{itemize}
\item[{\rm (i)}] 
$c^{-1} < G_n(v) < c$ \, ; 
\item[{\rm (ii)}] 
$|G_n(v) - G_n(v')| \leq C |v - v'|$ \, .
\end{itemize}
\label{asn:g}
\end{assumption}

In practice, as discussed earlier on $Q^{\infty}$ is typically analytically intractable (and we further suppose is not currently known up-to a non-negative unbiased estimate).
As a result,
we will focus upon targets associated to a discretization, i.e.~of the type
\begin{align}\label{eq:target_l}
\hat{\eta}^{l}_{n}(y_{1:n}) \propto \Big[\prod_{i=1}^{n}G_i(y_{i})Q^{l}(y_{i-1}, y_i)\Big], 
\end{align}
for $l<\infty$, where $Q^l$ is defined by $k_l$ iterates of the recursion in \eqref{eq:euler_levyd}.
Note that we will use $\hat{\eta}^{l}_{n}$ as the notation for measure and density, 
with the use clear from the context, where $l = 0,1,\dots, \infty$.  

The objective is to compute the expectation of functionals with respect to this measure, particularly at the last co-ordinate.
For any bounded and measurable function $f:\mathbb{R}^d \rightarrow\mathbb{R}$, $n\geq 1$, we will use the notation
\begin{align}\label{levy:eq12}
\hat{\eta}_{n}^{l}(f)&:=\int_{\bbR^{dn}}f(y_{n})\hat{\eta}^{l}_{n}(y_{1:n})\mathrm{d}y_{1:n}.
\end{align}

Often of interest is the computation of the un-normalized measure. 
 That is, for any bounded and measurable function $f:\mathbb{R}^d \rightarrow\mathbb{R}$
define, for $n\geq 1$
\begin{equation}\label{eq:marg_like}
\hat{\zeta}^{l}_{n}(f) := \int_{\mathbb{R}^{dn}}f(y_n) \Big[\prod_{i=1}^{n}G_i(y_{i})Q^{l}(y_{i-1}, y_i)\Big]\mathrm{d}y_{1:n}.
\end{equation}
In the context of the model under study, $\hat{\zeta}^{l}_{n}(1)$ is the marginal likelihood.

\red{Henceforth $Y^l_{1:n}$ will be used to denote a draw from $\hat{\eta}^{l}_{n}$.  
The vanilla case described earlier can be viewed as the special example in which $G_i\equiv1$ for all $i$. 
Following standard practice, realizations of random variables will be denoted with small letters.  
So, after drawing $Y^{l,(i)}_n\sim \hat{\eta}^l_n$, 
then the notation $y^{l,(i)}_n$ will be used for later references to the realized value.  
The randomness of the samples will be recalled again for MSE calculations, over potential realizations.}



\subsection{Particle Filtering}\label{levy:PF}

We will describe the particle filter that is capable of exactly approximating, that is as the Monte Carlo samples go to infinity, 
terms of the form \eqref{levy:eq12}  and \eqref{eq:marg_like}, for any fixed $l$. The particle filter 
has been studied and used extensively (see for example \cite{delm:04,doucet:01}) in many practical applications of interest.

For a given level $l$, algorithm \ref{algo:particlefilter} gives the standard particle filter.
The weights are defined as for $k\geq 1$ 
\begin{align}\label{levy:eq13}
w^{l,(i)}_{k} 
& =  w^{l,(i)}_{k-1} \frac{G_k(y^{l,(i)}_{k})}{\sum_{j=1}^{N_l}w_{k-1}^{l,(j)} G_k(y^{l,(j)}_{k})}
\end{align}
with the convention that $w^{l,(i)}_{0}=1$.
Note that the abbreviation $ESS$ stands for effective sample size which measures the variability of weights at time $k$ of the algorithm (other more efficient procedures are also possible,
but not considered). In the analysis to follow $H=1$ in  algorithm \ref{algo:particlefilter} (or rather it's extension in the next section), but this is not the case in our numerical implementations.

\cite{delm:04} (along with many other authors) have shown that 
\red{for upper-bounded, non-negative,  $G_i$, $f:\mathbb{R}^d\rightarrow\mathbb{R}$ bounded measurable (these conditions can be relaxed),}
at step 3 of algorithm \ref{algo:particlefilter}, the estimate
$$
\sum_{i=1}^{N_l} w^{l,(i)}_{n}f(y^{l,(i)}_{n})
$$
will converge almost surely to \eqref{levy:eq12}.  
In addition, if $H=1$ in  algorithm \ref{algo:particlefilter},
$$
\Big[\prod_{i=1}^{n-1}\frac{1}{N_l}\sum_{j=1}^{N_l}G_i(y^{l,(j)}_{i})\Big]\frac{1}{N_l}\sum_{j=1}^{N_l}G_n(y^{l,(j)}_{n})f(y^{l,(j)}_{n})
$$
will converge almost surely to \eqref{eq:marg_like}.

\begin{algorithm}[!ht]
\caption{\textbf{: Particle filter}}
\label{algo:particlefilter}
\begin{algorithmic}
\STATE
\begin{enumerate}
	\item[0.]  Set $k=1$; for $i=1,\dots, N_l$, draw $Y_1^{l,(i)}\sim Q^l(y_0.)$
	\item[1.]  Compute weights $\{w_1^{l,(i)}
	\}_{i=1}^{N_l}$ using $\left(\ref{levy:eq13}\right)$
	\item[2.]  Compute $ESS=\Big(\sum_{i=1}^{N_l}(w_k^{l,(i)})^{2}\Big)^{-1}$.\\  
	If {$ESS/N_l<H$} (for some threshold $H$), resample the particles $\{Y_{k}^{l,(i)} \}_{i=1}^{N_l}$ 
	and set all weights to $w_k^{l,(i)}=1/N_l$. Denote the resampled particles $\{\hat{Y}_{k}^{l,(i)} \}_{i=1}^{N_l}$.\\
	Else set $\{\hat{Y}_{k}^{l,(i)} \}_{i=1}^{N_l}=\{Y_{k}^{l,(i)} \}_{i=1}^{N_l}$
	\item[3.]  Set $k=k+1$; if $k=n+1$ stop;\\
	for $i=1\dots,N_l$, draw $Y_{k}^{l,(i)}\sim Q^l(\hat{y}_{k-1}^{l,(i)},.)$;\\
	compute weights $\{w_k^{l,(i)}\}_{i=1}^{N_l}$ by using $\left(\ref{levy:eq13}\right)$.
	Go to 2.
\end{enumerate}
\end{algorithmic}
\end{algorithm}

\subsection{Multilevel Particle Filter}\label{levy:mlpf} 

We now describe the multilevel particle filter of \cite{Jay_mlpf} for the context considered here.
The basic idea is to run $L+1$ independent algorithms, the first a particle filter as in the previous section and the remaining, 
coupled particle filters. The particle filter will sequentially (in time) approximate $\hat{\eta}_k^0$ and the coupled filters
 will sequentially approximate the couples $(\hat{\eta}^0_k,\hat{\eta}^1_k),\dots,(\hat{\eta}^{L-1}_k,\hat{\eta}^{L}_k)$.
Each (coupled) particle filter will be run with $N_l$ particles.  

The most important step in the MLPF is the coupled resampling step, which maximizes the probability of resampled indices being the same at the coarse
and fine levels. Denote the coarse and fine particles at level $l\geq 1$ and step $k\geq 1$ as $\Big(Y_{k}^{l,(i)}(l),Y_{k}^{l-1,(i)}(l)\Big)$, for $i=1,\dots,N_l$.
Equation \eqref{levy:eq13} is replaced by the following, for $k\geq 1$ 
\begin{align}\label{levy:mlweights}
w^{l,(i)}_{k}(l) 
& =  w^{l,(i)}_{k-1}(l) \frac{G_k(y^{l,(i)}_{k}(l))}{\sum_{j=1}^{N_l}w_{k-1}^{l,(j)}(l) G_k(y^{l,(j)}_{k}(l))} \\
w^{l-1,(i)}_{k}(l) 
& =  w^{l-1,(i)}_{k-1}(l) \frac{G_k(y^{l-1,(i)}_{k}(l))}{\sum_{j=1}^{N_l}w_{k-1}^{l-1,(j)}(l) G_k(y^{l-1,(j)}_{k}(l))} 
\end{align}
with the convention that $w^{l,(i)}_{0}(l)=w^{l-1,(i)}_{0}(l)=1$.

\begin{algorithm}[!ht]
\caption{\textbf{Coupled Resampling Procedure}}
\label{algo:coupresamp}
\begin{algorithmic}
\STATE For $\ell=1,\dots, N_l$
\STATE With probability $\sum_{i=1}^{N_l} \min\{ w^{l,(i)}_{k}(l), w^{l-1,(i)}_{k}(l) \}$,
\begin{enumerate} 
	\item[(i)] Sample $J$ with probability proportional to $\min\{ w^{l,(i)}_{k}(l), w^{l-1,(i)}_{k}(l) \}$ for $i=1, \dots, N_l$, 
	where the weights are computed according to  \eqref{levy:mlweights}.  
	\item[(ii)] Set $\Big (\hat{Y}^{l,(\ell)}_{k}(l), \hat{Y}^{l-1,(\ell)}_{k}(l)\Big)=\Big(Y_{k}^{l,(j)}(l),Y_{k}^{l-1,(j)}(l)\Big)$.
\end{enumerate}
\STATE
\STATE Else, with probability $1-\sum_{i=1}^{N_l} \min\{ w^{l,(i)}_{k}(l), w^{l-1,(i)}_{k}(l) \}$,
\begin{enumerate} 
	\item[(i)] Sample $J_l$ with probability proportional to $w^{l,(i)}_{k}(l) - \min\{ w^{l,(i)}_{k}(l), w^{l-1,(i)}_{k}(l) \}$ for $i=1, \dots, N_l$,   
	\item[(ii)] Sample $J_{l-1}\perp J_l$ with probability proportional to $w^{l-1,(i)}_{k}(l) - \min\{ w^{l,(i)}_{k}(l), w^{l-1,(i)}_{k}(l) \}$ for $i=1, \dots, N_l$,   
	\item[(iii)] Set $\hat{Y}^{l,(\ell)}_{k}(l)=Y_{k}^{l,(j_l)}(l)$, and $\hat{Y}^{l-1,(\ell)}_{k}(l)=Y_{k}^{l-1,(j_{l-1})}(l)$.
\end{enumerate}
\end{algorithmic}
\end{algorithm}

In the below description, we set $H=1$ (as in algorithm \ref{algo:particlefilter}), but it need not be the case. Recall that the case $l=0$
is just a particle filter. For each $1\leq l \leq L$ the following procedure is run independently.\\

\begin{algorithm}[!ht]
\caption{\textbf{Multilevel Particle filter}}
\label{algo:particlefilter}
\begin{algorithmic}
\STATE
\begin{enumerate}
	\item[0.]  Set $k=1$; for $i=1,\dots, N_l$, draw $\Big(Y_{1}^{l,(i)}(l),Y_{1}^{l-1,(i)}(l)\Big) {\sim}M^{l}\Big((y_{0},y_{0}), .\Big)$. 
	\item[1.]  Compute weights $\{ (w_1^{l,(i)}(l), w_1^{l-1,(i)}(l)) \}_{i=1}^{N_l}$ using $\left(\ref{levy:mlweights}\right)$
	\item[2.]  Compute $ESS=\min\Big \{ \Big(\sum_{i=1}^{N_l}(w_k^{l,(i)}(l))^{2}\Big)^{-1},  \Big(\sum_{i=1}^{N_l}(w_k^{l-1,(i)}(l))^{2}\Big)^{-1}\Big \}$.\\  
	If {$ESS/N_l<H$}, resample the particles $\Big\{\Big(\hat Y_{k}^{l,(i)}(l),\hat Y_{k}^{l-1,(i)}(l)\Big) \Big\}_{i=1}^{N_l}$ 
	according to Algorithm \ref{algo:coupresamp}, 
	and set all weights to $w_k^{l,(i)}(l)=w_k^{l-1,(i)}(l)=1/N_l$. 
	Else set $\Big\{\Big(\hat Y_{k}^{l,(i)}(l),\hat Y_{k}^{l-1,(i)}(l)\Big) \Big\}_{i=1}^{N_l} = \Big\{\Big(Y_{k}^{l,(i)}(l),Y_{k}^{l-1,(i)}(l)\Big) \Big\}_{i=1}^{N_l}$
	\item[3.]  Set $k=k+1$; if $k=n+1$ stop;\\
	for $i=1\dots,N_l$, draw $\Big(Y_{k}^{l,(i)}(l),Y_{k}^{l-1,(i)}(l)\Big) {\sim}M^{l}\Big((\hat{y}_{k-1}^{l,(i)}(l),\hat{y}_{k-1}^{l-1,(i)}(l)), .\Big)$; \\ 
	compute weights $\{(w_k^{l,(i)}(l), w_k^{l-1,(i)}(l))\}_{i=1}^{N_l}$ by using \eqref{levy:mlweights}.  Go to 2.
\end{enumerate}
\end{algorithmic}
\end{algorithm}

%
%
The samples generated by the particle filter for $l=0$ at time $k$ are denoted $Y_{k}^{0,(i)}(0)$, $i\in\{1,\dots,N_0\}$ (we are assuming {$H=1$}).

To estimate the quantities \eqref{levy:eq12} and \eqref{eq:marg_like} (with $l=L$) \cite{Jay_mlpf,Jay_mlnormconst} 
show that in the case of discretized diffusion processes 
$$
\hat{\eta}^{\rm ML,L}_n(f) =  \sum_{l=1}^L 
\Big(
\frac{\sum_{i=1}^{N_l}G_n(y_n^{l,(i)}(l))f(y_n^{l,(i)}(l))}{\sum_{i=1}^{N_l}G_n(y_n^{l,(i)}(l))} -
\frac{\sum_{i=1}^{N_l}G_n(y_n^{l-1,(i)}(l))f(y_n^{l-1,(i)}(l))}{\sum_{i=1}^{N_l}G_n(y_n^{l-1,(i)}(l))}
\Big) +
$$
$$
\frac{\sum_{i=1}^{N_0}G_n(y_n^{0,(i)}(0))f(y_n^{0,(i)}(0))}{\sum_{i=1}^{N_0}G_n(y_n^{0,(i)}(0))}
$$
and
$$
\hat{\zeta}^{\rm ML,L}_n(f) =\sum_{l=1}^L 
\Big(
\Big[\prod_{i=1}^{n-1}\frac{1}{N_l}\sum_{j=1}^{N_l}G_i(y^{l,(j)}_{i}(l))\Big]\frac{1}{N_l}\sum_{j=1}^{N_l}G_n(y^{l,(j)}_{n}(l))f(y^{l,(j)}_{n}(l)) -
$$
$$
\Big[\prod_{i=1}^{n-1}\frac{1}{N_l}\sum_{j=1}^{N_l}G_i(y^{l-1,(j)}_{i}(l))\Big]\frac{1}{N_l}\sum_{j=1}^{N_l}G_n(y^{l-1,(j)}_{n}(l))f(y^{l-1,(j)}_{n}(l))
\Big) 
$$
\begin{equation}\label{eq:nc_est_ml}
+
\Big[\prod_{i=1}^{n-1}\frac{1}{N_0}\sum_{j=1}^{N_0}G_i(y^{0,(j)}_{i}(0))\Big]\frac{1}{N_0}\sum_{j=1}^{N_0}G_n(y^{0,(j)}_{n}(0))f(y^{0,(j)}_{n}(0))
\end{equation}
converge almost surely to $\hat{\eta}^L_n(f)$ and $\hat{\zeta}^{L}_n(f)$ respectively, 
as min$\{N_l\} \rightarrow \infty$. Furthermore, both can significantly improve over the particle filter, 
for $L$ and $\{N_l\}_{l=1}^L$ appropriately chosen to depend upon a target mean square error (MSE). 
By improve, we mean that 
the work is less than the particle filter to achieve a given MSE with respect to the continuous time limit, 
under appropriate assumptions on the diffusion. 
We show how the $N_0,\dots,N_L$ can be chosen in Section \ref{sec:theorem}.  
Note that for positive $f$ the estimator above $\hat{\zeta}^{\rm ML,L}_n(f)$ can take negative values with positive probability.

We remark that the coupled resampling method can be improved as in \cite{Sen_coupling}. 
We also remark that the approaches of \cite{hous_ml, jacob} could potentially be used here. However,
none of these articles has sufficient supporting theory to verify a reduction in cost of the ML procedure.

\subsubsection{Theoretical Result}\label{sec:theorem}

We conclude this section with 
a technical theorem.  We consider only $\hat{\eta}^{\rm ML,L}_n(f)$, 
but this can be extended to $\hat{\zeta}^{\rm ML,L}_n(f)$, similarly to \cite{Jay_mlnormconst} .
The proofs are given in Appendix \ref{app:theo}.

Define $\mathcal{B}_b(\mathbb{R}^d)$ as the bounded, measurable and real-valued functions
on $\mathbb{R}^d$ and $\textrm{Lip}(\mathbb{R}^d)$ as the globally Lipschitz real-valued functions on $\mathbb{R}^d$. 
Define the space 
$\cA=\mathcal{B}_b(\mathbb{R}^d)\cap\textrm{Lip}(\mathbb{R}^d)$
with the norm 
$\|\varphi\| = \sup_{x\in \mathbb{R}^d} |\varphi(x)| + \sup_{x,y \in \mathbb{R}^d} \frac{|\varphi(x) - \varphi(y)|}{|x-y|}$.

%
%


\red{The following assumptions will be required.
\begin{assumption}\label{asn:delh}
For all $h_l>0$, there exists a solution $\delta_l(h_l)$ to the equation 
$h_l = 1/\nu(B_{\delta_l(h_l)}^c)$, and some $C,\beta_1>0$ such that 
$\delta_l(h_l) \leq C h_l^{\beta_1}$.
\end{assumption}

Denote by $\check{Q}^{l,l-1}((y,y'),\cdot)$ the coupling of the Markov transitions $Q^l(y,\cdot)$ and $Q^{l-1}(y',\cdot)$, $(y,y')\in\mathbb{R}^{2d}$ as in Algorithm \ref{levy:coupledkernelAlgo}.
\begin{assumption}
There is a $\gamma>0$ such that
\begin{itemize}
\item 
$\bbE[{\rm COST}(\check{Q}^{l,l-1})] 
= \cO(h_l^{-\gamma})$,
\end{itemize}
where $\bbE[{\rm COST}(\check{Q}^{l,l-1})]$ 
is the cost to simulate one sample from the kernel $\check{Q}^{l,l-1}$.
\label{asn:mlrates}
\end{assumption} }

Below $\mathbb{E}$ denotes expectation w.r.t.~the law of the particle system.
\begin{theorem} \label{thm:mlpf}
Assume (\ref{ass:main}, \ref{asn:g},\ref{asn:delh}, \ref{asn:mlrates}). 
Then for any $n\geq 0$, there exists a $C<+\infty$ 
such that for
 $\varepsilon>0$ given and a particular $L>0$, and $\{N_l\}_{l=0}^L$ depending upon $\varepsilon$, $h_{0:L}$ only
and
$f \in \cA$, 
\[
\mathbb{E}\Bigg[
\Bigg(\hat{\eta}^{\rm ML,L}_n(f)
- \hat{\eta}_n^\infty(f) \Bigg)^2   \Bigg] \leq
C
\varepsilon^2, 
\]
for the cost $\cC(\varepsilon) := \bbE[{\rm COST}(\varepsilon)]$ 
given in the second column of Table \ref{tab:mlpfcases}.

\begin{table}[h]
\begin{center}
  \begin{tabular}{ | c || c | c |}
    \hline
    CASE & 
    $\cC(\varepsilon)$ \\ \hline\hline
    $\beta>2\gamma$ & 
    $\cO(\varepsilon^{-2})$ \\ \hline
    $\beta=2\gamma$  & 
    $\cO(\varepsilon^{-2}\log(\varepsilon)^2)$ \\ \hline
    $\beta<2\gamma$  & 
    $\cO(\varepsilon^{-2+(\beta-2\gamma)/(\beta)})$ \\
    \hline
  \end{tabular}
\end{center}
\caption{The three cases of MLPF, and associated 
cost $\cC(\varepsilon)$. $\beta$ is as Lemma \ref{prop:uni}}
\label{tab:mlpfcases}
\end{table}

\end{theorem}

\begin{proof}
The proof is essentially identical to \cite[Theorem 4.3]{Jay_mlpf}. The only
difference is to establish analogous results to \cite[Appendix D]{Jay_mlpf}; this is
done in the appendix of this article.
\end{proof}

%

\section{Numerical Examples}\label{levy:numerics}

In this section, we compare our proposed multilevel particle filter method with the vanilla particle filter method.  
A target accuracy parameter $\epsilon$ will be specified 
and the cost to achieve an error below this target accuracy will be estimated.  
The performance of the two algorithms will be compared in two applications of SDEs driven by general L\'{e}vy process:   
filtering of a partially observed L\'{e}vy process (S\&P $500$ stock price data) and pricing of a path dependent option.  
In each of these two applications, we let $X=\{X_t\}_{t\in[0,K]}$ denote a symmetric stable L\'{e}vy process, 
i.e. $X$ is a $\left(\nu,\Sigma, b \right)$-L\'{e}vy process, and Lebesgue density of the L\'{e}vy measure given by
\begin{align}\label{levy:eq15}
\nu(\mathrm{d}x)&=c|x|^{-1-\phi}\bbone_{[-x^*,0)}(x)\mathrm{d}x+
c|x|^{-1-\phi}\bbone_{(0,x^*]}(x)\mathrm{d}x,\quad x\in\mathbb{R}\setminus\{0\},
\end{align}
with $c>0$, $x^*>0$ (the truncation threshold) and index $\phi\in(0,2)$.  
The 
parameters $c$ and $x^*$ are both $1$ for all the examples considered. 
The L\'{e}vy-driven SDE considered here has the form
\begin{align}\label{levy:eq16}
\mathrm{d}Y_t=a(Y_{t^{-}})\mathrm{d}X_t,\quad Y_{0}=y_{0}, 
\end{align}
with 
$y_0$ assumed known, and $a$ satisfies Assumption \ref{ass:main}(i).  
Notice Assumption \ref{ass:main}(ii-iii) are also satisfied by the L{\'e}vy process defined above.
In the examples illustrated below, we take $a(Y_t)=Y_t, y_0=1$, and $\phi=0.5$. 


\begin{remark}[Symmetric Stable L\'{e}vy process of index $\phi\in(0,2)$]
\label{levy:stable}
	
	In approximating the L\'{e}vy-driven SDE $\left(\ref{levy:eq16}\right)$, Theorem $2$ of \cite{Dereich_mlmcLevydriven} provided 
	asymptotic error bounds for the strong approximation by the Euler scheme.  
	If the driving L\'{e}vy process $X_t$ has no Brownian component, that is $\Sigma=0$, 
	then the $L^2$-error, denoted $\sigma_{h_l}^2$, is bounded by
	\begin{align*}
	\sigma_{h_l}^2&\leq C(\sigma^2(\delta_l)+{ |b-F_0^l|^2}h_l^2),
	\end{align*}
	and for $\Sigma > 0$,
	\begin{align*}
	\sigma_{h_l}^2&\leq C(\sigma^2(\delta_l)+h_l|\log(h_l)|),
	\end{align*}
	for a fixed constant $C<\infty$ (that is the Lipschitz constant), 
	where $\sigma^2(\delta_l) := \int_{B_{\delta_l}} |x|^2 \nu(dx)$. 
	Recall that $\delta_l(h_l)$ is chosen such that $h_l=1/{\nu(B_{\delta_l}^{c})}$.
	One obtains the analytical expression
	\begin{align}\label{levy:eq17}
	\sigma^2(\delta_l) & = \frac{2c}{2-\phi}\delta_l(h_l)^{2-\phi}\leq C\delta_l^{2-\phi},
	\end{align}
	for some constant $C>0$.  
	One can also analytically compute 
    \begin{align*}
	\nu(B_{\delta_l}^{c})&=\frac{2c(\delta_l^{-\phi}-{x^*}^{-\phi})}{\phi}.
    \end{align*}
    Now, setting $h_l=2^{-l}$, one obtains
	\begin{align}\label{levy:eq18}
	\delta_l&=\Big(\frac{2^{l}\phi}{2c}+{x^*}^{-\phi}\Big)^{-1/\phi},
	\end{align}
	so that the L\'{e}vy measure $\nu(B_{\delta_l}^{c})=2^l$, 
	\red{hence verifying assumption \ref{asn:delh} for this example.}  
	Then, one can easily bound $\left(\ref{levy:eq18}\right)$ by
	\begin{align*}
	|\delta_l|\leq C2^{-l/\phi}
	\end{align*}
	for some constant $C>0$.  
	So $\delta_l=\mathcal{O}(h_{l}^{1/\phi})$.  Using $\left(\ref{levy:eq17}\right)$-$\left(\ref{levy:eq18}\right)$ and the error bounds for $\Sigma=0$, one can straightforwardly obtain strong error rates for the approximation of SDE driven by stable L\'{e}vy process in terms of the single accuracy parameter $h_l$.  This is given by
	\begin{align*}
	\sigma_{h_l}^2&\leq C (h_l^{(2-\phi)/\phi}+|b-F_0^l|^2h_l^{2}).
	\end{align*}
	{Thus, if $b-F_0^l \neq 0$, 
	the strong error rate $\beta$ of Assumption \ref{asn:mlrates}(ii) associated with a particular discretization level $h_l$ is given by 
	\begin{align}\label{levy:eq19}
	\beta & = \min\Big(\frac{2-\phi}{\phi},2\Big).
	\end{align}
	Otherwise it is just given by $(2-\phi)/\phi$.
	}
\end{remark}

In the examples considered below, the original L\'{e}vy process has no drift and Brownian motion components, that is $\Sigma=b=0$.  Due to the linear drift correction $F_0^l$ in the compensated compound Poisson process, the random jump times are refined such that the time differences between successive jumps are bounded by the accuracy parameter $h_l$ associated with the Euler discretization approximation methods in $\left(\ref{levy:eq4}\right)$ and $\left(\ref{levy:eq9}\right)$-$\left(\ref{levy:eq10}\right)$.  However, since $F_0^l=0$ here, due to symmetry, this does not affect the rate, as described in Remark \ref{levy:stable}.

We start with verification of the weak and strong error convergence rates, $\alpha$ and $\beta$ for the forward model.  
To this end the quantities $|\mathbb{E}[Y^{l}_1-Y^{l-1}_1]|$ and $\mathbb{E}[|Y^{l}_1-Y^{l-1}_1|^2]$ are computed over increasing levels $l$.  
Figure \ref{levy:fig1} shows these computed values plotted against $h_l$ on base-$2$ logarithmic scales.  
A fit of a linear model gives rate $\alpha=1.3797$, and similar simulation experiment gives $\beta=2.7377$.
This is consistent with the rate $\beta=3$ and $\alpha=\beta/2$ 
from Remark \ref{levy:stable} $\left(\ref{levy:eq19}\right)$.

\begin{figure}[!ht]
	\centering
	\subfigure{{\includegraphics[height=7cm]{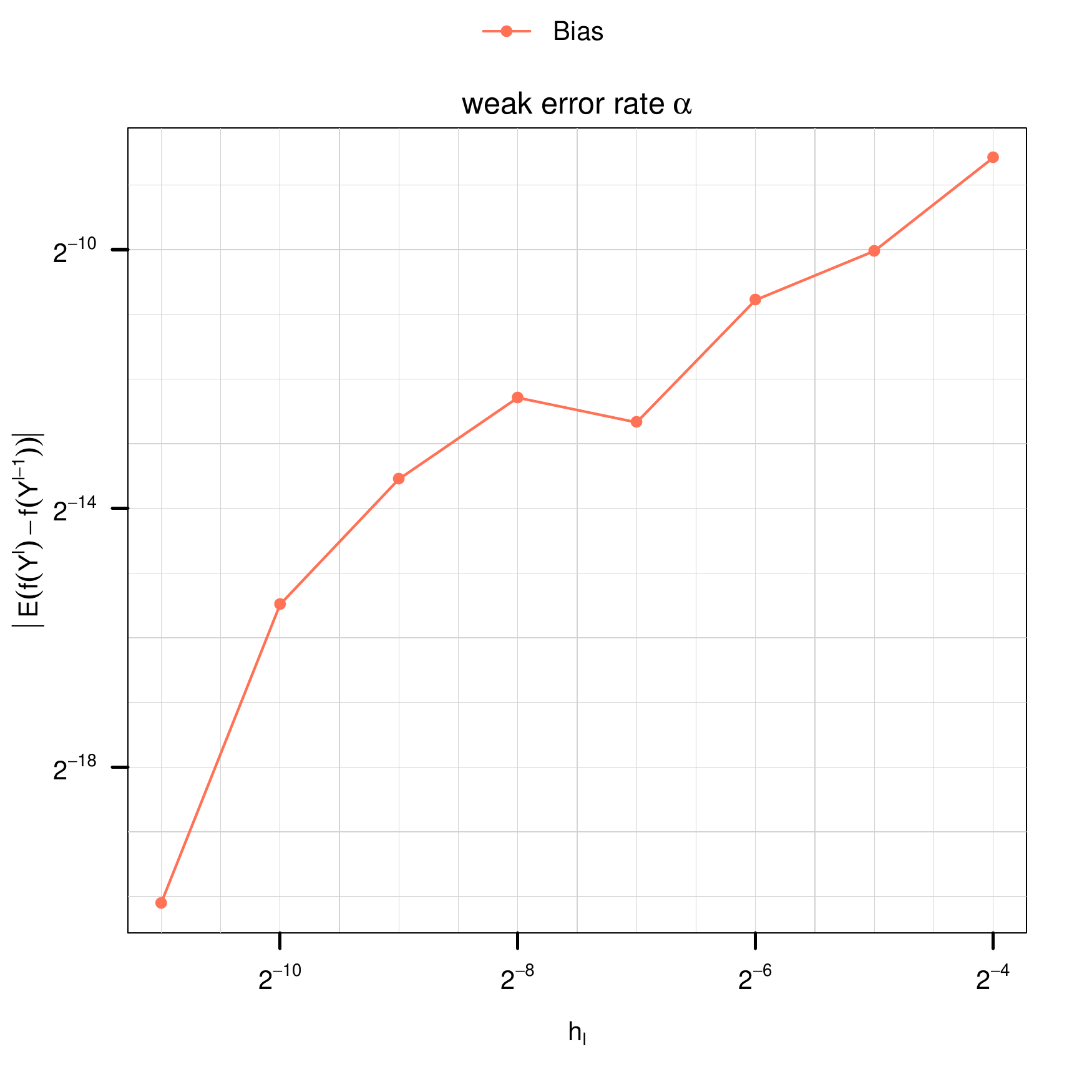}}}\,\,
	\subfigure{{\includegraphics[height=7cm]{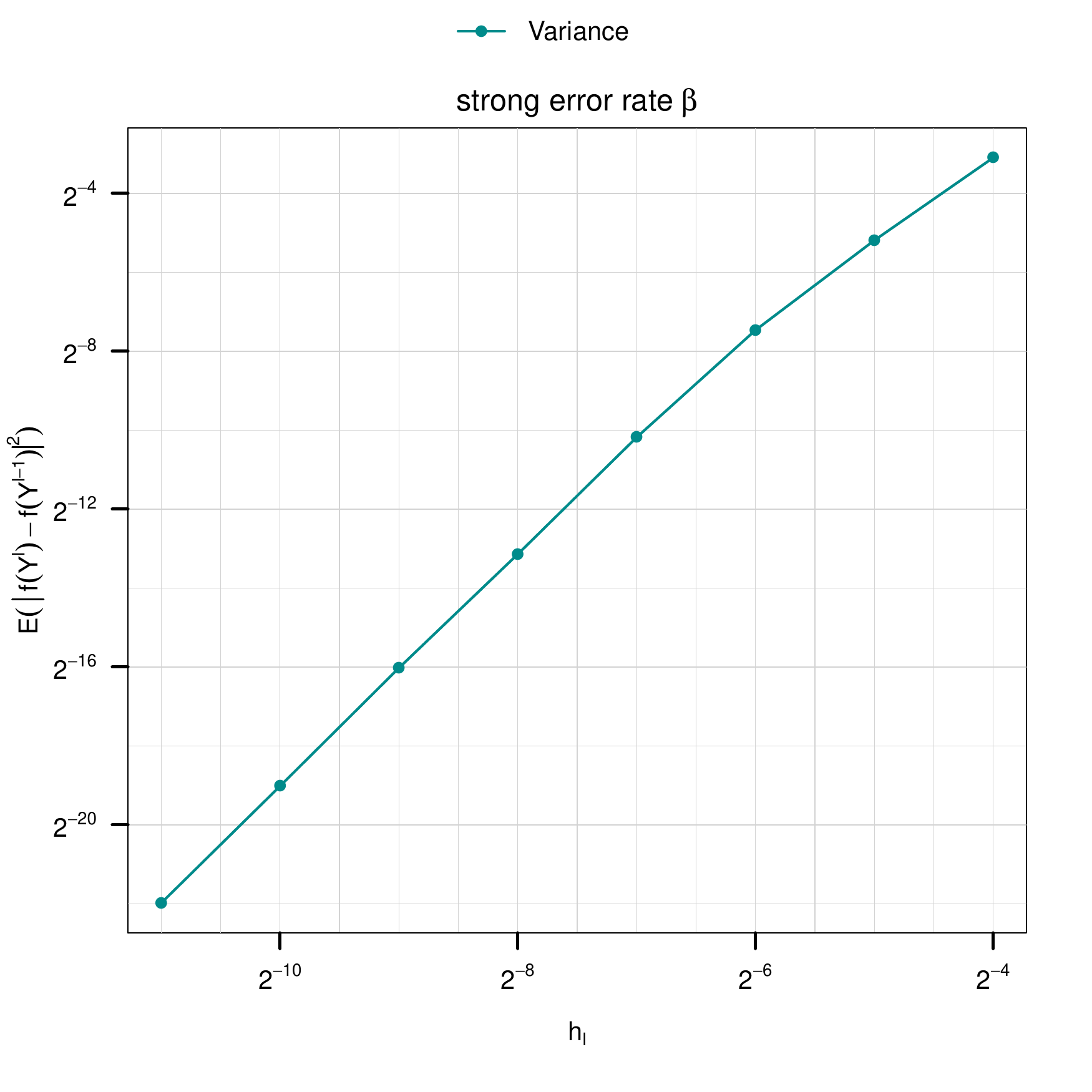}}}
	\caption{Empirical weak and strong error rates estimates}
	\label{levy:fig1}
\end{figure}

We begin our comparison of the MLPF and PF algorithms starting with the filtering of a partially observed L\'{e}vy-driven SDE and then consider the knock out barrier call option pricing problem.

\subsection{Partially observed data}

In this section 
we 
consider filtering a partially observed L\'{e}vy process.  
Recall that the L\'{e}vy-driven SDE takes the form $\left(\ref{levy:eq16}\right)$.  
In addition, partial observations $\{z_1,\dots,z_n\}$ are available with $Z_k$ obtained at time $k$ and 
$Z_k|(Y_{k}=y_k)$ has a density function $G_k(y_{k})$ (with observation is omitted from the notation and appearing only as subscript $k$).  
The observation density is Gaussian with mean $y_k$ and variance 1.
We aim to estimate $\mathbb{E}[f(Y_{k})|z_{1:k}]$ for some test function $f(y)$.  In this application, we consider the real daily S\&P $500$ $\log$ return data (from August $3$, $2011$ to July $24$, $2015$, normalized to unity variance).  
We shall take the test function $f(y)=e^{y}$ for the example considered below, which we note does not satisfy the assumptions of Theorem \ref{thm:mlpf}, and hence challenges the theory. In fact the results are roughly equivalent to the case 
$f(y)=e^{y}\mathbb{I}_{\{|y|<10\}}$, where $\mathbb{I}_A$ is the indicator function on the set $A$, which was also considered and does satisfy the required assumptions.

\begin{figure}[!ht]
	\centering
	\subfigure{{\includegraphics[height=8cm,width=12cm]{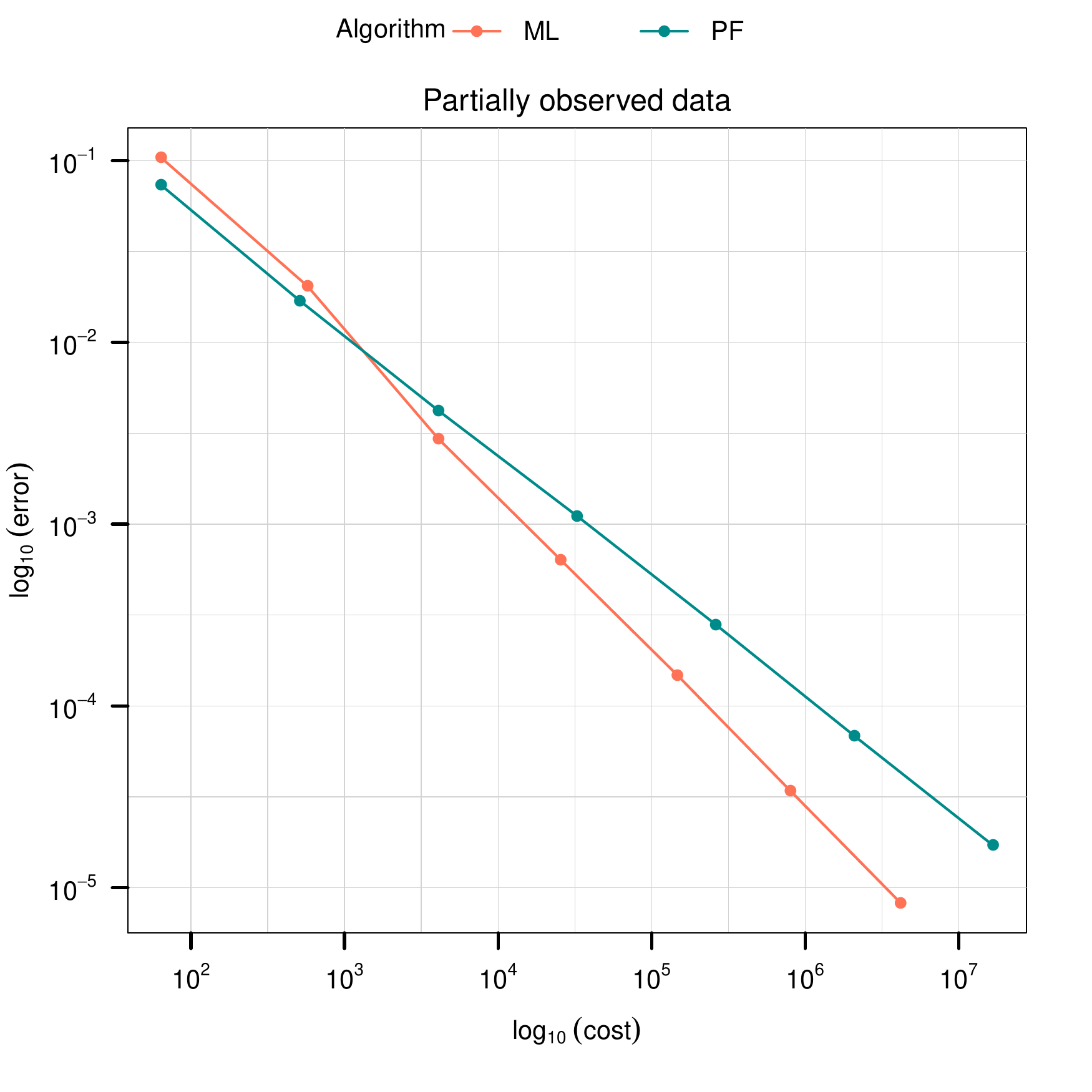}}}\,
	\caption{Mean square error against computational cost for filtering with partially observed data.}
	\label{levy:fig3}
\end{figure}

The error-versus-cost plots on base $10$ logarithmic scales for PF and MLPF are shown in Figure \ref{levy:fig3}.  The fitted linear model of $\log$ MSE against $\log$ Cost has a slope of $-0.667$ and $-0.859$ for PF and MLPF respectively.  These results again verify numerically the expected theoretical asymptotic behaviour of computational cost as a function of MSE for both standard cost and ML cost.  


\subsection{Barrier Option}

Here we consider computing the value of a discretley monitored knock out barrier option (see e.g.~\cite{glasserman} and the references therein). Let $Y_0\in[a,b]$ for some $0<a<b<+\infty$ known and let $Q^{\infty}(y_{i-1},y)$ be the transition density of the process 
as in \eqref{levy:eq16}. Then the value of the barrier option (up-to a known constant) is
$$
\int_{\bbR^{n}} \max\{y_n-S,0\}\prod_{i=1}^n \mathbb{I}_{[a,b]}(y_i)Q^{\infty}(y_{i-1},y_i) dy_{1:n}
$$
for $S>0$ given.
As seen in \cite{JayOption} the calculation of the barrier option is non-trivial, in the sense that even importance sampling may not work well. We consider the (time) discretized version
$$
\int_{\bbR^{n}} \max\{y_n-S,0\}\prod_{i=1}^n \mathbb{I}_{[a,b]}(y_i)Q^{l}(y_{i-1},y_i) dy_{1:n}.
$$
Define a sequence of probability densities, $k\in\{1,\dots,n\}$
\begin{equation}\label{eq:barpos}
\hat{\eta}_k^l(y_{1:k}) \propto \tilde{G}_k(y_k)\prod_{i=1}^k \mathbb{I}_{[a,b]}(y_i)Q^{l}(y_{i-1},y_i) = \prod_{i=1}^k 
\left( \frac{\tilde{G}_i(y_i)}{\tilde{G}_{i-1}(y_{i-1})} \right) \mathbb{I}_{[a,b]}(y_i) Q^{l}(y_{i-1},y_i)
\end{equation}
for some non-negative collection of functions $\tilde{G}_k(y_k)$, $k\in\{1,\dots,n\}$ to be specified.
Recall that $\hat{\zeta}_n^l$ denotes the un-normalized density associated to $\hat{\eta}_n^l$.
Then the value of the time discretized barrier option is exactly
\begin{equation}\label{eq:barnc}
\hat{\zeta}_{n}^l\Big(\frac{f}{\tilde{G}_n}\Big) = \int_{\bbR^{n}} \max\{y_n-S,0\}\prod_{i=1}^n \mathbb{I}_{[a,b]}(y_i)Q^{l}(y_{i-1},y_i) dy_{1:n}
\end{equation}
where $f(y_n)=\max\{y_n-S,0\}$. Thus, we can apply the MLPF targetting the sequence $\{\hat{\eta}_k^l\}_{ k\in\{1,\dots,n\},l \in\{0,\dots,L\} }$ 
and use our normalizing constant estimator \eqref{eq:nc_est_ml} to estimate \eqref{eq:barnc}. 
\red{If $\tilde{G}_n=|f|$, then we have an optimal importance distribution, 
in the sense that we are estimating the integral of the constant function $1$ 
and the variance is minimal \cite{rubinstein}.  
However, noting the form of the effective potential above \eqref{eq:barpos},
this can result in infinite weights (with adaptive resampling as done here), and so some regularization is necessary.
We bypass this issue by choosing 
$\tilde{G}_k(y_k) =  |y_k-S|^{\kappa_k}$, where $\kappa_k$ is an annealing parameter with $\kappa_0 = 0$ and $\kappa_n =  1$.
We make no claim that this is the best option, but it guides us to something reminiscent of the optimal thing, 
and with well-behaved weights, in practice.  We tried also $\max\{y_n-S,\varepsilon\}$, with $\varepsilon=0.001$, 
and the results are almost identical.}

For this example we choose $S=1.25,a=0,b=5,y_0=1,n=100$. 
The $N_l$ are chosen as in the previous example.
The error-versus-cost plots for PF and MLPF are shown in Figure \ref{levy:fig2}.  Note that the bullets in the graph correspond to different choices of $L$ (for both PF and MLPF, $2\leq L\leq8$).   The fitted linear model of $\log$ MSE against  $\log$ cost  has a slope of $-0.6667$ and $-0.859$ for PF and MLPF respectively.  These numerical results are consistent with the expected theoretical asymptotic behaviour of 
MSE$\propto$Cost$^{-1}$ for the multilevel method.  
The single level particle filter achieves the asymptotic behaviour of the standard Monte Carlo method with 
MSE$\propto$Cost$^{-2/3}$. 

\begin{figure}[!ht]
	\centering
	\subfigure{{\includegraphics[height=8cm,width=12cm]{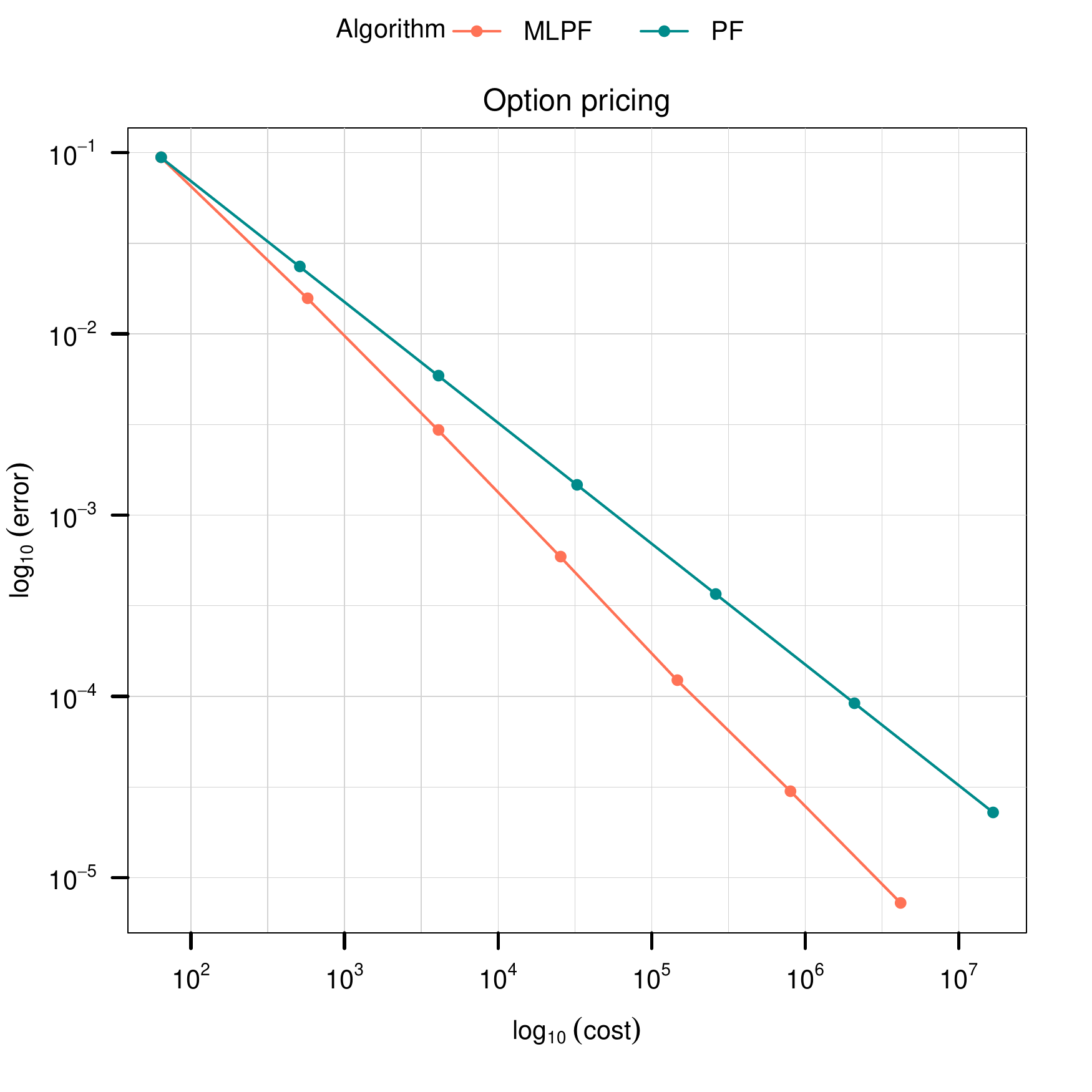}}}\,
	\caption{Mean square error against computational cost for the knock out barrier option example.}
	\label{levy:fig2}
\end{figure}

%
%

\subsubsection*{Acknowledgements}
AJ was supported by Singapore ministry of education AcRF tier 2 grant R-155-00-161-112 and he is affiliated with the CQF, RMI and ORA cluster at NUS.  He was also supported by a King Abdullah University of Science and Technology Competitive Research Grant round 4, Ref:2584.  KJHL was sponsored by the Laboratory Directed Research and Development Program of Oak Ridge National Laboratory, managed by UT-Battelle, LLC, for the U. S. Department of Energy.

\appendix

\section{Theoretical results}
\label{app:theo}

Our proof consists of following the proof of \cite{Jay_mlpf}. To that end
all the proofs of \cite[Appendices A-C]{Jay_mlpf} are the same for the approach
in this article (note that one needs Lemma \ref{lem:lip_cont_Q} of this article along the way). 
One must verify the analogous results of \cite[Appendix D]{Jay_mlpf},
which is what is done in this appendix.

The predictor at time $n$, level $l$, is denoted as $\eta^l_{n}$. 
Denote the 
total variation norm 
as $\|\cdot\|_{\textrm{tv}}$.
 For $\varphi\in\textrm{Lip}(\bbR^d)$, $\|\varphi\|_{\textrm{Lip}}:= \sup_{x,y \in \mathbb{R}^d} \frac{|\varphi(x) - \varphi(y)|}{|x-y|}$ is the Lipschitz constant.
For ease (in abuse) of notation, $Q^l$ defined by $k_l$ iterates of the recursion in \eqref{eq:euler_levyd} is used as a Markov kernel below.
We set for $\varphi\in\mathcal{B}_b(\bbR^d)$, $y\in\mathbb{R}^d$
$$
Q^l(\varphi)(y) := \int_{\mathbb{R}^d} \varphi(y')Q^l(y,y')dy'.
$$
Recall, for $l\geq 1$, $\check{Q}^{l,l-1}((y,y'),\cdot)$ is the coupling of the kernels $Q^l(y,\cdot)$ and $Q^{l-1}(y',\cdot)$ as in Algorithm \ref{levy:coupledkernelAlgo}.
For $\varphi\in\mathcal{B}_b(\bbR^{2d})$ we use the notation for $(y,y')\in\mathbb{R}^{2d}$:
$$
\check{Q}^{l,l-1}(\varphi)(y,y') := \int_{\mathbb{R}^{2d}} \varphi(y^l,y^{l-1})\check{Q}^{l,l-1}((y,y'),d(y^l,y^{l-1}))
$$
and note that for $\varphi\in\mathcal{B}_b(\bbR^d)$
$$
\check{Q}^{l,l-1}(\varphi\otimes 1)(y,y')  = Q^l(\varphi)(y), \quad \check{Q}^{l,l-1}(1 \otimes \varphi)(y,y')  = Q^{l-1}(\varphi)(y')
$$
where $\otimes$ denotes the tensor product of functions, e.g.~$\varphi\otimes 1$ denotes $\varphi(y^l)$ in the integrand associated
to $\check{Q}^{l,l-1}((y,y'),d(y^l,y^{l-1}))$.

\red{
Let $T_j(t) = \max \{T_j \in \bbT ; T_j < t \}$, 
and let $(\Delta X)^l_{t}  = X^l_t - X^l_{T_j(t)}$,
where $X^l_t$ is the natural continuation of the 
discretized L\'evy process \eqref{levy:eq4}.
Define the continuation of the discretized driven process by
$$
Y^l_t= Y^l_{T_j(t)} + a(Y^l_{T_j(t)}) (\Delta X)^l_{t} \, .
$$
}
Let $Y^l_1\sim Q^l(y,\cdot)$ and independently $Y^{'l}_1\sim Q^l(y',\cdot)$. We denote expectations w.r.t.~these random variables as $\mathbb{E}$.

\begin{lem}\label{prop:lip_strong}
Assume (\ref{ass:main}).  
Then there exists a $C<+\infty$ such that for any $L\geq l\geq 0$, 
and $(y,y')\in \bbR^{2d}$
$$
\bbE |Y^l_1 - Y^{'l}_1|^2  \leq 
C | y - y' |^2 \, .
$$
\end{lem}

\begin{proof}
Let $t\in [0,1]$. We have 
\[\begin{split} 
|Y^l_t - Y^{'l}_t|^2  
= & ~~ |Y^l_{T_j(t)} - Y^{'l}_{T_j(t)}|^2 + 2 \left (Y^l_{T_j(t)} - Y^{'l}_{T_j(t)} \right)^T 
\left (a(Y^l_{T_j(t)}) - a(Y^{'l}_{T_j(t)}) \right) (\Delta X)^l_{t}  \\
 & + \left | \left(a(Y^l_{T_j(t)}) - a(Y^{'l}_{T_j(t)})\right) (\Delta X)^l_{t}\right|^2  \, .
\end{split}\]

Let $N = \# \{T_j \leq 1\}$ be the number of time-steps
before time $1$, 
and denote $\bbT = \{\tilde{T}_1,\dots, \tilde{T}_{N}, N\}$,
where $\tilde{T}_j$ and $T_j$ are generated by Algorithm \ref{levy:DiscreteAlgo}.
The sigma algebra generated by these random variables is denoted $\sigma(\bbT)$.

Following from the independence of $Y^l_{T_j(t)}$ and  $(\Delta X)^l_{t}$ conditioned on $\sigma(\bbT)$, we have
\[\begin{split} 
\bbE \left[ |Y^l_t - Y^{'l}_t|^2 \Big | \sigma(\bbT) \right]  
 & \leq  \bbE \left[ |Y^l_{T_j(t)} - Y^{'l}_{T_j(t)}|^2 \Big | \sigma(\bbT) \right]  \\
 & + 
2 \bbE \left[ \left (Y^l_{T_j(t)} - Y^{'l}_{T_j(t)} \right)^T 
\left (a(Y^l_{T_j(t)}) - a(Y^{'l}_{T_j(t)}) \right)  \Big | \sigma(\bbT) \right] \bbE \left[ (\Delta X)^l_{t}   | \sigma(\bbT) \right] 
\end{split}\]
\begin{equation}\label{eq:la11}
+ \bbE \left[  \left | a(Y^l_{T_j(t)}) - a(Y^{'l}_{T_j(t)}) \right |^2  \Big | \sigma(\bbT) \right] \bbE \left[ |(\Delta X)^l_{t}|^2 \Big | \sigma(\bbT) \right].
\end{equation}
The inequality is a result of the last term which uses the definition of the matrix 2 norm.
Note that $\bbE[W_t ]= \bbE[L_t] = 0$, so that $\bbE[(\Delta X)^l_{t} | \sigma(\bbT)] = (b - F_0^l) (t - T_j(t)) $.
In addition, \eqref{eq:efl}, Jensen's inequality and the fact that $B^c_{\delta_l} \subset B_{\delta_l}$, 
together imply that 
\begin{equation}\label{eq:eflin}
| F_0^l |^2 \leq \int_{B_{\delta_l}^c} |x|^2 \nu(dx) \leq \int |x|^2 \nu(dx) \, .
\end{equation}

We have 
$$
\bbE \left[ \left (Y^l_{T_j(t)} - Y^{'l}_{T_j(t)} \right)^T 
\left (a(Y^l_{T_j(t)}) - a(Y^{'l}_{T_j(t)}) \right)  \Big | \sigma(\bbT) \right] \bbE \left[ (\Delta X)^l_{t}   | \sigma(\bbT) \right] 
$$
$$
= \bbE \left[ \left (Y^l_{T_j(t)} - Y^{'l}_{T_j(t)} \right)^T 
\left (a(Y^l_{T_j(t)}) - a(Y^{'l}_{T_j(t)}) \right) \Big | \sigma(\bbT) \right]  (b - F_0^l) (t - T_j(t))
$$
\begin{equation}\label{eq:l12}
\leq C^2 h_l \bbE \left[ \left |Y^l_{T_j(t)} - Y^{'l}_{T_j(t)} \right |^2 
 \Big | \sigma(\bbT) \right]  
\end{equation}
The inequality follows from Cauchy-Schwarz, 
definition of the matrix 2 norm, 
Assumption \ref{ass:main}(i), (iii), and (ii) in connection with \eqref{eq:eflin} and 
the definition of the construction of $\{T_j\}$ in Algorithm \ref{levy:DiscreteAlgo}, 
so that $|t-T_j(t)| \leq h_l$.

%
%


Note also 
\begin{equation}\label{eq:square}
\bbE \left[ |(\Delta X)^l_{t}|^2 \Big | \sigma(\bbT) \right] \leq C^2( |t-T_j(t)| + |t-T_j(t)|^2 ) \leq  C^2h_l\, , 
\end{equation}
by Assumption \ref{ass:main} (ii) and (iii), and since 
$h_l\leq 1$ by definition.
Returning to \eqref{eq:la11}, and using \eqref{eq:square} and \eqref{eq:l12}, 
and Assumption \ref{ass:main} (i) again on the last term, we have
\begin{equation}\label{eq:fin1step}
\bbE \left[ |Y^l_t - Y^{'l}_t|^2 \Big | \sigma(\bbT) \right] \leq \bbE \left[ |Y^l_{T_j(t)} - Y^{'l}_{T_j(t)}|^2 \Big | \sigma(\bbT) \right] (1 + Ch_l) \, ,
\end{equation}
where the value of the constant is different.

Therefore, in particular
$$
\bbE \left[ |Y^l_{T_{j+1}} - Y^{'l}_{T_{j+1}} |^2 \Big | \sigma(\bbT) \right] \leq \bbE \left[ |Y^l_{T_j} - Y^{'l}_{T_j}|^2 \Big | \sigma(\bbT) \right] (1 + Ch_l) \, .
$$
By applying \eqref{eq:fin1step} recursively, we have
$$
\bbE \left[ |Y^l_1 - Y^{'l}_1|^2 \Big | \sigma(\bbT) \right] \leq 
| y - y' |^2 (1+Ch_l)^{N}   \, .
$$
Note that 
$\bbP(N=n) = \frac{(\lambda^{\delta_l})^n}{n!} e^{-\lambda^{\delta_l}}$,
and $\lambda^{\delta_l}=h_l^{-1}$ by design, as described in Section \ref{numApprox}.
Taking expectation with respect to $\sigma(\bbT)$ gives 
\[\begin{split}
\bbE |Y^l_1 - Y^{'l}_1|^2 & \leq 
\left ( \sum_{n\geq 0} 
\frac{(h_l^{-1} (1+Ch_l) )^{n} }{n!} e^{-h_l^{-1}} \right )
| y - y' |^2  \\
& = e^{C} | y - y' |^2  \, .
\end{split}
\]
The result follows by redefining $C$.
\end{proof}

\begin{lem}\label{lem:lip_cont_Q}
Assume (\ref{ass:main}). Then there exists a 
$C<+\infty$ such that for any $L\geq l\geq 0$, $(y,y')\in \bbR^{2d}$, and $\varphi\in\mathcal{B}_b(\bbR^d)\cap\textrm{\emph{Lip}}(\bbR^d)$
$$
|Q^l(\varphi)(y) - Q^l(\varphi)(y')| \leq C\|\varphi\|_{\textrm{\emph{Lip}}}~|y-y'|.
$$
\end{lem}
\begin{proof}
We have
\begin{eqnarray*}
|Q^l(\varphi)(y) - Q^l(\varphi)(y')|
& = & | \bbE (\varphi(Y^l_1) - \varphi(Y^{'l}_1))| \\
& \leq & (\bbE |\varphi(Y^l_1) - \varphi(Y^{'l}_1)|^2)^{1/2} \\
& \leq & \|\varphi\|_{\textrm{Lip}} (\bbE |Y^l_1 - Y^{'l}_1|^2)^{1/2}\, 
\end{eqnarray*}
where Jensen has been applied to go to the second line and that 
$\varphi\in\textrm{Lip}(\bbR^d)$ to the third. The proof is concluded via  Lemma \ref{prop:lip_strong}. 
\end{proof}


\begin{lem}\label{prop:uni}
Assume (\ref{ass:main}, \ref{asn:delh}). 
Then there exists $C<+\infty$ such that for any $L\geq l\geq 1$
$$
\sup_{\varphi\in\mathcal{A}}\sup_{y\in\bbR^d}|Q^l(\varphi)(y)-Q^{l-1}(\varphi)(y)| \leq C h_{l}^{\frac{\beta}{2}}.
$$
\end{lem}

\begin{proof}
We have
$$
|Q^l(\varphi)(y)-Q^{l-1}(\varphi)(y)| = \Big|\int_{\mathbb{R}^{2d}}\varphi(y^l) \check{Q}^{l,l-1}((y,y),d(y^l,y^{l-1})) - 
\int_{\mathbb{R}^{2d}}\varphi(y^{l-1}) \check{Q}^{l,l-1}((y,y),d(y^l,y^{l-1}))\Big|.
$$
Using Jensen's inequality 
yields
$$
|Q^l(\varphi)(y)-Q^{l-1}(\varphi)(y)| \leq \Big(\int_{\mathbb{R}^{2d}}(\varphi(y^l) -\varphi(y^{l-1}))^2 \check{Q}^{l,l-1}((y,y),d(y^l,y^{l-1}))\Big)^{1/2}.
$$
Recall for $\varphi \in \cA$ there exists $C<+\infty$ such that $|\varphi(y^l) -\varphi(y^{l-1})| \leq C |y^l -y^{l-1}|$.
By \cite[Theorem 2]{Dereich_mlmcLevydriven}, there exists $C<+\infty$ such that for any $y\in\mathbb{R}^d$, $l\geq 1$
\begin{equation}\label{eq:strong}
\int_{\mathbb{R}^{2d}}|y^l -y^{l-1}|^2 \check{Q}^{l,l-1}((y,y),d(y^l,y^{l-1})) \leq C h_l^\beta \, 
\end{equation}
The proof is then easily concluded.
\end{proof}

\begin{rem}
\red{To verify \eqref{eq:strong}  note that $\delta_l(h_l)$ is chosen as a function of $h_l$ here for simplicity, 
and by Assumption \ref{asn:delh} 
it can be bounded by $C h_l^{\beta_1}$ for some $\beta_1$.
The bounds in Theorem 2 of \cite{Dereich_mlmcLevydriven} can therefore be written as 
the sum of two terms $C(h_l^{\beta_1} + h_l^{\beta_2})$, 
and $\beta = \min\{\beta_1,\beta_2\}$.}  
See remark \ref{levy:stable} for calculation of $\beta$ in the example considered in this paper.  
\end{rem}


\begin{lem}\label{prop:strong_cor}
Assume (\ref{ass:main},\ref{asn:delh}).
Then there exists $C<+\infty$ and $\beta>0$ such that for any $L\geq l\geq 1$, and $(y,y')\in \bbR^{2d}$, 
$$
\Big(\int_{\mathbb{R}^{2d}}|y^l -y^{l-1}|^2 \check{Q}^{l,l-1}((y,y'),d(y^l,y^{l-1}))\Big)^{1/2}
\leq C (|y-y'| +  h_l^{\beta/2}) \,  
$$
where $\beta$ is as in Lemma \ref{prop:uni}.
\end{lem}

\begin{proof} 
We have
$$
\Big(\int_{\mathbb{R}^{2d}}|y^l -y^{l-1}|^2 \check{Q}^{l,l-1}((y,y'),d(y^l,y^{l-1}))\Big)^{1/2} = 
$$
$$
\Big(\int_{\mathbb{R}^{3d}} |y^l -\bar{y}^l +\bar{y}^l -y^{l-1}|^2 \check{Q}^{l,l-1}((y,y'),d(y^l,y^{l-1}))
Q^l(y',d\bar{y}^l)\Big)^{1/2} \leq
$$
$$
\Big(\int_{\mathbb{R}^{2d}} |y^l -\bar{y}^l|^2 Q^{l}(y,dy^l) Q^l(y',d\bar{y}^l)\Big)^{1/2} + 
$$
$$
\Big(\int_{\mathbb{R}^{2d}}|\bar{y}^l -y^{l-1}|^2Q^{l}(y',d\bar{y}^l) Q^{l-1}(y',dy^{l-1}) \Big))\Big)^{1/2} \leq 
$$
$$
C| y - y' |
+
\Big(\int_{\mathbb{R}^{2d}}|\bar{y}^l -y^{l-1}|^2Q^{l}(y',d\bar{y}^l) Q^{l-1}(y',dy^{l-1})\Big)^{1/2}
$$
where we have applied Minkowski's inequality to go to the third line and Lemma \ref{prop:lip_strong} to go to the final line.
Now
$$
\Big(\int_{\mathbb{R}^{2d}}|\bar{y}^l -y^{l-1}|^2Q^{l}(y',d\bar{y}^l) Q^{l-1}(y',dy^{l-1})\Big)^{1/2} = 
$$
$$
\Big(\int_{\mathbb{R}^{3d}}|\bar{y}^l- \tilde{y}^l+ \tilde{y}^l -y^{l-1}|^2Q^{l}(y',d\bar{y}^l)
\check{Q}^{l,l-1}((y',y'),d(\tilde{y}^l,y^{l-1}))\Big)^{1/2} \leq
$$
$$
 \Big(\int_{\mathbb{R}^{2d}}|\bar{y}^l- \tilde{y}^l|^2Q^l(y',d\bar{y}^l)Q^l(y',d\tilde{y}^l)\Big)^{1/2} + 
\Big(\int_{\mathbb{R}^{2d}} |\tilde{y}^l -y^{l-1}|^2\check{Q}^{l,l-1}((y',y'),d(\tilde{y}^l,y^{l-1}))\Big)^{1/2}
$$
where again we have applied Minkowski's inequality to go to the third line. Then
$$
\int_{\mathbb{R}^{2d}}|\bar{y}^l- \tilde{y}^l|^2Q^l(y',d\bar{y}^l)Q^l(y',d\tilde{y}^l) = 0
$$
and by \eqref{eq:strong} we have
$$
\Big(\int_{\mathbb{R}^{2d}}|\bar{y}^l -y^{l-1}|^2Q^{l}(y',d\bar{y}^l) Q^{l-1}(y',dy^{l-1})\Big)^{1/2}\leq C h_l^{\frac{\beta}{2}}.
$$
The argument is then easily concluded. 
\end{proof}

\begin{proposition}\label{prop:tv}
Assume (\ref{ass:main},\ref{asn:g},\ref{asn:delh}).
Then there exists a $C<+\infty$ such that for any $L\geq l\geq 1$, $n\geq 0$, 
\begin{equation}\label{eq:tv}
\|\eta^l_{n} - \eta^{l-1}_{n}\|_{\rm tv} \leq C h_l^{\frac{\beta}{2}} \, .
\end{equation}
where $\beta$ is as Lemma \ref{prop:uni}.
\end{proposition}

\begin{proof}
The result follows from the same calculations of the proof of \cite[Lemma D.2]{Jay_mlpf}
along with our Lemma \ref{prop:uni}, which we note is analogous to (32) in \cite{Jay_mlpf} with $\alpha=\beta/2$
\end{proof}

It is remarked that, given our above results, Lemmata D.3 and D.4 as well as Theorem D.5 (all of \cite{Jay_mlpf})
can be proved for our algorithm by the same arguments as in \cite{Jay_mlpf} and are hence omitted.

Note that we have proved that:
\begin{eqnarray*}
\sup_{\varphi\in\mathcal{A}}\sup_{y\in\bbR^d}|Q^l(\varphi)(y)-Q^{l-1}(\varphi)(y)| & \leq & C h_{l}^{\frac{\beta}{2}} \, ,\\
\Big|\int_{\mathbb{R}^{2d}}\varphi(y^l) \check{Q}^{l,l-1}((y,y),d(y^l,y^{l-1})) -  \int_{\mathbb{R}^{2d}}\varphi(y^{l-1}) \check{Q}^{l,l-1}((y,y),d(y^l,y^{l-1}))\Big| & \leq & C h_{l}^{\frac{\beta}{2}} \, ,\\
\int_{\mathbb{R}^{2d}}(\varphi(y^l) -\varphi(y^{l-1}))^2 \check{Q}^{l,l-1}((y,y'),d(y^l,y^{l-1})) 
& \leq & C  h_l^{\beta} \, ,
\end{eqnarray*}
for all $\varphi \in \cA$.
This provides \cite[Assumption 4.2.~(i) \& (ii)]{Jay_mlpf}, with $\alpha$ (as in \cite{Jay_mlpf}) equal to $\beta/2$.

\end{document}